\documentclass[11pt]{article}
\usepackage[utf8]{inputenc}

\usepackage{graphicx}
\usepackage{xcolor}
\usepackage{amsmath,amsthm,amssymb}
\usepackage[margin=1in]{geometry}
\usepackage[pdftex,colorlinks=true,linkcolor=blue,citecolor=blue,urlcolor=blue]{hyperref}
\usepackage{comment}
\usepackage[normalem]{ulem}
\usepackage{cite}
\usepackage{tcolorbox}

\def\ba{\begin{array}}
\def\ea{\end{array}}

\def\0{{\bf 0}}
\def\a{{\bf a}}
\def\b{{\bf b}}

\def\u{{\bf u}}
\def\U{{\mathbb U}}

\newcommand{\be}{\begin{equation}}
\newcommand{\ee}{\end{equation}}
\newcommand{\bea}{\begin{eqnarray}}
\newcommand{\eea}{\end{eqnarray}}
\newcommand{\bes}{\begin{equation*}}
\newcommand{\ees}{\end{equation*}}
\newcommand{\beas}{\begin{eqnarray*}}
\newcommand{\eeas}{\end{eqnarray*}}

\makeatletter
\newtheorem*{rep@theorem}{\rep@title}
\newcommand{\newreptheorem}[2]{%
\newenvironment{rep#1}[1]{%
 \def\rep@title{#2 \ref{##1} (restated)}%
 \begin{rep@theorem}}%
 {\end{rep@theorem}}}
\makeatother

\allowdisplaybreaks

\newtheorem{thm}{Theorem}
\newtheorem*{thm*}{Theorem}

\newtheorem{con}[thm]{Conjecture}
\newtheorem{lem}[thm]{Lemma}
\newtheorem*{lem*}{Lemma}

\newtheorem{prop}[thm]{Proposition}

\newtheorem{rem}[thm]{Remark}

\newreptheorem{thm}{Theorem}
\newreptheorem{lem}{Lemma}
\newreptheorem{cor}{Corollary}

\usepackage{pgfplots}

\usepackage{authblk}

\title{ Joint moments of higher order derivatives of CUE characteristic polynomials I: asymptotic formulae }

\author[1]{Jonathan P. Keating\thanks{keating@maths.ox.ac.uk}}
\author[2]{Fei Wei\thanks{weif@mail.tsinghua.edu.cn
}}
\affil[1]{Mathematical Institute,
University of Oxford, Oxford OX2 6GG, UK}
\affil[2]{Yau Mathematical Sciences Center, Tsinghua University}

\date{\today}

\begin{document}

\maketitle

\begin{abstract}

We derive explicit asymptotic formulae for the joint moments of the $n_1$-th and $n_2$-th derivatives of the characteristic polynomials of CUE random matrices for any non-negative integers $n_1, n_2$. These formulae are expressed in terms of determinants whose entries involve modified Bessel functions of the first kind. We also express them in terms of two types of combinatorial sums. Similar results are obtained for the analogue of Hardy's $Z$-function. We use these formulae to formulate general conjectures for the joint moments of the $n_1$-th and $n_2$-th derivatives of the Riemann zeta-function and of Hardy's $Z$-function. Our conjectures are supported by comparison with results obtained previously in the number theory literature. 

\end{abstract}


\section{Introduction}

Deep and remarkable connections exist between random matrix theory and number theory.  These originated in the work of Montgomery \cite{montgomery1973pair}, who conjectured that suitably normalised pairs of zeros of the Riemann zeta-function have a statistical distribution that coincides with Dyson's  \cite{dyson1962statistical} description of spacings between pairs of eigenvalues of random unitary matrices. This connection was strengthened by the work of Keating and Snaith \cite{keating2000random1,keating2000random2}, which used characteristic polynomials of random unitary matrices to model the Riemann zeta-function, leading to general conjectures, supported by results from number theory and by numerical simulations, concerning the value distribution and moments of the zeta function and other $L$-functions. These connections and conjectures have been widely studied in the past few years, e.g., see the survey papers \cite{keating2017random, akemann2011oxford, snaith2010riemann}.

Let $A\in \U(N)$ be taken from the Circular Unitary Ensemble (CUE) of random matrices, and let $\Lambda_{A}(s)$ be the characteristic polynomial of $A$, defined by
\[
\Lambda_A(s) = \prod_{n=1}^N (1-se^{-\i\theta_n}),
\]
where we set $\i^2=-1$ to avoid confusion with the index $i$, and  $e^{\i\theta_1}, \ldots,e^{\i\theta_{N}}$ are the eigenvalues of $A$. Define
\[
Z_{A}(s) := e^{-\pi \i N/2}e^{\i \sum_{n=1}^{N}\theta_{n}/2}
s^{- N/2}\Lambda_{A}(s),
\]
where for $s^{-N/2}$, when $N$ is an odd integer, we use the branch of the square root that is positive for $s$ real and positive.
This is the analogue of Hardy's $Z$-function. 
It satisfies the functional equation 
$Z_A(s) = (-1)^N Z_{A^*}(1/s)$, where $A^*$ is the conjugate transpose of $A$. This implies that $Z_A(e^{\i \theta})$ is real when $\theta$ is real. 

In this paper, we study the joint moments of  the $n_1$-th and $n_2$-th derivatives of the characteristic polynomial
\be
\label{intro:eq1}
\int_{\U(N)} |\Lambda_A^{(n_1)}(1)|^{2M} |\Lambda_A^{(n_2)}(1)|^{2k-2M}  dA_N
\ee
and
\be
\label{intro:eq2}
\int_{\U(N)} |Z_A^{(n_1)}(1)|^{2M} |Z_A^{(n_2)}(1)|^{2k-2M} dA_N, 
\ee
where $n_1,n_2 \in \mathbb{N}$ are non-negative integers denoting derivatives with respect to the variables in question, and $dA_N$ is the Haar measure on $\U(N)$.

When $n_1=1, n_2=0$, the joint moments have been studied extensively. We here list some results that are most relevant to ours in that case.
In \cite{hughes2001characteristic}, Hughes showed that the limits as $N\to\infty$ of $(\ref{intro:eq1})/N^{k^2+2M}$ and
$(\ref{intro:eq2})/N^{k^2+2M}$ 
exist when $k$ and $M$ are integers, and provided expressions for them that are analytic in $k$ for ${\rm Re}(k) > M - 1/2$.
In \cite{conrey2006moments}, Conrey, Rubinstein, and Snaith showed that when $M=k$ (both integers)
\[
(\ref{intro:eq1}) \sim (-1)^{\frac{k(k+1)}{2}} N^{k^2+2k}
\sum_{h=0}^k \binom{k}{h} 
\left( \frac{d}{dx} \right)^{k+h} \left( e^{-x} x^{-{k^2}/{2}} \det_{k\times k} \Big(I_{i+j-1}(2\sqrt{x}) \Big) \right)\Bigg|_{x=0},
\]
where $I_\alpha(x)$ is the modified Bessel function of the first kind, and the indices $i, j$ in the above determinant are from 1 to $k$ (we will assume this for all $k\times k$ determinants in this paper). They also gave a similar asymptotic formula for (\ref{intro:eq2}) when $M=k$.
In \cite{bailey2019mixed}, Bailey, et. al. provided asymptotic ($N\to\infty$) formulae for (\ref{intro:eq1}), (\ref{intro:eq2}) when $k\geq M \geq 0$ (both integers). In \cite{forrester2006boundary}, Forrester and Witte provided alternate asymptotic expressions for (\ref{intro:eq1}), 
(\ref{intro:eq2}) in terms of solutions of $\sigma$-Painlev\'{e} III$'$ equation. 
In \cite{winn2012derivative}, Winn gave an explicit expression for (\ref{intro:eq2}) in terms of certain sums over partitions and also gave an explicit expression for non-integer $M$ of the form $(2m-1)/2$. In \cite{basor2019representation}, Basor, et. al. established an exact representation of (\ref{intro:eq2}) in terms of a solution of $\sigma$-Painlev\'{e} V equation and showed how this could be used to derive the large-$N$ asymptotics.  In \cite{AssiotisKeatingWarren}, Assiotis, Keating and Warren succeeded in proving that the limits as $N\to\infty$ of $(\ref{intro:eq1})/N^{k^2+2M}$ and
$(\ref{intro:eq2})/N^{k^2+2M}$ exist for all real $k$ and $M$ in the appropriate ranges.  We refer to \cite{dehaye2008joint,assiotis2021distinguished,forrester2022joint,assiotis2022convergence} for other related results.

For one of $n_1, n_2$ greater than 1, to the best of our knowledge, there are no previous results for (\ref{intro:eq2}).
In \cite[section 10]{dehaye2008joint} Dehaye mentioned that his methods might be applied to (\ref{intro:eq1}) for some $n_1, n_2$, but he did not give details.
In \cite[Theorem 4.31]{barhoumi2020new}, Barhoumi-Andr\'{e}ani obtained an asymptotic formula for (\ref{intro:eq1}) and for the more general moments (\ref{general case of 1}). He obtains the leading coefficient in the form of a $(k-1)$-fold real integral of a certain function. In comparison, we provide two types of explicit formulae in the form of combinatorial sums, see Theorems \ref{intro:prop1} and \ref{intro:thm3}, which presumably  correspond to evaluation of his integral.
Our approach is quite different from Dehaye's and Barhoumi-Andr\'{e}ani's. Our starting point is the Conrey-Farmer-Keating-Rubinstein-Snaith formula (see Lemma \ref{11101}),
whereas the method mentioned in \cite{dehaye2008joint} is based on expressing (\ref{intro:eq1}) in the basis of Schur functions, and that set out in \cite{barhoumi2020new} is based on replacing Lemma \ref{11101} by certain Fourier representations (see (54), (55) and Section 1.5 in \cite{barhoumi2020new}) obtained by analysing the reproducing kernel on the space of all symmetric functions.
Barhoumi-Andr\'{e}ani argues that his equations (54) and (55) can be represented in the form of Hankel determinants in Section 5.1.3 in \cite{barhoumi2020new}. So it is likely that Barhoumi-Andr\'{e}ani's result has a link to certain Hankel determinants. In Theorem \ref{intro:thm1}, using our approach mentioned above, we give an explicit formula for the leading coefficient in terms of determinants of Hankel matrices whose columns are shifted by integers, and so which are not  Hankel determinants anymore. Moreover, this shift is important in building the connection between moments of higher-order derivatives of CUE characteristic polynomials and a solution of the $\sigma$-Painlev\'{e} III$'$ equation, which we establish in \cite{keating-fei}.

We remark that when defining $\widetilde{\Lambda}_A(\theta):=\Lambda_A(e^{\i \theta})$ and $\widetilde{Z}_A(\theta):=Z_A(e^{\i \theta})$, one may consider (\ref{intro:eq1}), (\ref{intro:eq2}) with $\Lambda^{(n)}_A(1)$ replaced by $\widetilde{\Lambda}^{(n)}_A(0)$ and $Z^{(n)}_A(1)$ replaced by $\widetilde{Z}^{(n)}_A(0)$ respectively. However, using Lemma \ref{11/10lemma1} (by taking appropriate orders of derivatives with respect to the shifts and letting the shifts be zero), it is not hard to prove that (\ref{intro:eq1}), (\ref{intro:eq2})  are indeed the main terms in the asymptotics of the corresponding moments of $\widetilde{\Lambda}^{(n)}_A(0)$ and $\widetilde{Z}^{(n)}_A(0)$. Intuitively, $\widetilde{\Lambda}^{(n)}_A(0)$ is a linear combination of $\Lambda^{(i)}_A(1)$ for $i=1,\ldots,n$, and $\Lambda^{(i)}_A(1)$ contributes $N^i$. So the main term in $|\widetilde{\Lambda}^{(n)}_A(0)|^{2h}$ indeed comes from $|\Lambda^{(n)}_A(1)|^{2h}$. Because of this, it suffices to focus on $\Lambda^{(n)}_A(1)$ and $Z^{(n)}_A(1)$ in (\ref{intro:eq1}), (\ref{intro:eq2}).

\subsection{Main results}

The main contribution of this paper is to give asymptotic formulae for (\ref{intro:eq1}), (\ref{intro:eq2}) for general integers $n_1, n_2$ and general integers $k\geq M\geq 0$. More precisely, we show that
\beas
(\ref{intro:eq1}) &=& a_{k,M}(n_1,n_2) N^{k^2+2Mn_1+2(k-M)n_2} + O(N^{k^2+2Mn_1+2(k-M)n_2-1}), \\
(\ref{intro:eq2}) &=& b_{k,M}(n_1,n_2) N^{k^2+2Mn_1+2(k-M)n_2} + O(N^{k^2+2Mn_1+2(k-M)n_2-1})
\eeas
for some $a_{k,M}(n_1,n_2), b_{k,M}(n_1,n_2)$ depending on $k,M,n_1,n_2$ with explicit formulae.
We then use these asymptotic formulae to conjecture the joint moments of higher order derivatives of the Riemann zeta-function on the critical line and of Hardy's $Z$-function. The conjectures are supported by results in the number theory literature.

The general results for (\ref{intro:eq1}), (\ref{intro:eq2}) are given explicitly in Theorems 
\ref {518mixmomentofLamda}-\ref{general 1128mixed moment}. The expressions are long, so here we illustrate them by presenting some non-trivial special cases. Theorems \ref{intro:thm1} and \ref{into:thm2} are asymptotic formulae for (\ref{intro:eq1}), (\ref{intro:eq2}) when $N\to\infty$ for $n_1=2,n_2=0$.  They involve derivatives of determinants of modified Bessel functions of the first kind shifted by integers. Theorems \ref{intro:prop1} and \ref{into:prop2}  respectively provide an explicit description of the coefficients appearing in the main terms in Theorems \ref{intro:thm1} and \ref{into:thm2}, express them in terms of partitions.

\begin{thm}
\label{intro:thm1}
Let $k\geq 1$, $0\leq M\leq k$ be integers. Then we have
\be
\int_{\U(N)} |\Lambda_A''(1)|^{2M} |\Lambda_A(1)|^{2k-2M} dA_N
= a_{k,M}(2,0) N^{k^2+4M} + O(N^{k^2+4M-1}),
\ee
where
\bea
a_{k,M}(2,0)
&=& (-1)^{\frac{k(k-1)}{2}} 
\sum_{l_1=0}^{M} \sum_{l_2=0}^{M} \sum_{t=0}^{2l_1}
\sum_{ \substack{ h_{1}+ \cdots+ h_{k} \\ =\, 2M-l_1-l_2  }} \binom{M}{l_1} \binom{M}{l_2}  \binom{2l_1}{t} \binom{2M-l_1-l_2}{h_{1}, \ldots, h_{k}} \nonumber \\
&&\times \, 
\left( \frac{d}{dx} \right)^{2l_2+t} \left(e^{-x}x^{-\frac{k^2}{2} - 2M+ l_1+l_2}
\det_{k\times k} \Big( I_{2h_{i}+i+j-1 } (2\sqrt{x}) \Big) \right) \Bigg|_{x=0} .
\label{intro:thm1-eq}
\eea
\end{thm}

\begin{thm}
\label{intro:prop1}
\beas
a_{k,M}(2,0)
&=& (-1)^{\frac{k(k-1)}{2}}
\sum_{l_1=0}^{M} \sum_{l_2=0}^{M} \sum_{t=0}^{2l_1} \sum_{ \substack{h_{10}+h_{1,1}+\cdots+h_{1k} = 2l_2+t \\ h_{21}+\cdots+h_{2k} = 2M-(l_1+l_2) } }  
\binom{M}{l_1} \binom{M}{l_2}  \binom{2l_1}{t} \binom{2l_2+t}{h_{10}, \ldots, h_{1k}}  \nonumber \\
&& 
\hspace{-2cm}
\times \, \binom{2M-(l_1+l_2)}{h_{21}, \ldots, h_{2k}}(-1)^{h_{10}}\left(\prod_{i=1}^k \frac{1}{(2k+\sum_{s=1}^{2} sh_{si}-i)!} \right)   \prod_{1\leq i<j\leq k} \left(\sum_{s=1}^{2} sh_{sj} - \sum_{s=1}^{2} sh_{si}-j+i \right) . 
\eeas
\end{thm}

\begin{thm}
\label{into:thm2}
Let $k\geq 1$, $0\leq M\leq k$ be integers. Then we have
\be
\int_{\U(N)} |Z_A''(1)|^{2M} |Z_A(1)|^{2k-2M} dA_N
=  b_{k,M}(2,0) N^{k^2+4M}  + O(N^{k^2+4M-1}),
\ee
where
\[
b_{k,M}(2,0) = 
(-1)^{\frac{k(k-1)}{2}} 
\hspace{-.8cm} 
\sum_{\substack{h+l_1+\cdots+l_k=2M \\ h\geq 0, l_1\geq 0, \ldots,l_k\geq 0}}
\binom{2M}{h,l_1,\ldots,l_k} \left( \frac{d}{dx} \right)^{2h}  \left( e^{-\frac{x}{2}} x^{h-2M-\frac{k^2}{2} } \det_{k\times k}\left(I_{2l_i+i+j-1}(2\sqrt{x}) \right) \right) \Bigg|_{x= 0}.
\]
\end{thm}

\begin{thm}
\label{into:prop2}
\beas
b_{k,M}(2,0) 
&=& (-1)^{\frac{k(k-1)}{2}} \sum_{\substack{l_0+\cdots+l_k=2M \\ l_0\geq 0,\ldots,l_k\geq 0}}
\,\,\,
\sum_{\substack{m_0+\cdots+m_k=2l_0 \\ m_0\geq 0, \ldots,m_k\geq 0}}  (-2)^{-m_0} \binom{2M}{l_0,\ldots,l_k}
\binom{2l_0}{m_0,\ldots,m_k}
  \\
&& \times \, \Bigg(\prod_{i=1}^k \frac{1}{(2k-i+2l_i+m_i)!} \Bigg) \Bigg(\prod_{1\leq i<j\leq k} (2l_j-2l_i+m_j-m_i-j+i)\Bigg).
\eeas
\end{thm}

The following two theorems respectively provide alternative forms of the asymptotic formulae for (\ref{intro:eq1}), (\ref{intro:eq2}) when $n_1=n, n_2=0$. Compared to Theorems \ref{518mixmomentofLamda} and \ref{generalization of mixed moment} with $n_{1}=n$, $n_{2}=0$, they have a significant computational advantage when $k$ is small and $n$ is large. For more about the differences in the general case, we refer to Remark \ref{remark: the difference}.

\begin{thm}
\label{intro:thm3}
For $k\geq 1$, $0\leq M\leq k$ and $n\geq 0$ be integers, we have
\be
\int_{\U(N)} |\Lambda_A^{(n)} (1)|^{2M} |\Lambda_A(1)|^{2k-2M} dA_N=
a_{k,M}(n,0) N^{k^2+2Mn}  +O(N^{k^2+2Mn-1}),
\ee
where
\beas
a_{k,M}(n,0) &=&(-1)^{\frac{k(k-1)}{2}}
(n!)^{2M}
\sum_{ \substack{ \sum_{l=1}^k s_{l,q} \leq n \\ \sum_{l=1}^k h_{l,q} = n \\
q=1,\ldots,M} } 
(-1)^{\sum_{q=1}^{M}\sum_{l=1}^k s_{l,q}} 
\left(\prod_{q=1}^{M}\frac{1}{(n-\sum_{l=1}^{k}s_{l,q})!} \right) \\
&&
\hspace{-2cm}
\times \,
\left(\prod_{i=1}^k \frac{1}{(2k-i+\sum_{q=1}^{M} (s_{i,q}+h_{i,q}))!} \right) 
\prod_{1\leq i<j\leq k} \left(\sum_{q=1}^{M} (s_{j,q}+h_{j,q})-\sum_{q=1}^{M} (s_{i,q}+h_{i,q})-j+i \right).
\eeas
\end{thm}

\begin{thm}
\label{intro:thm4}
Let $k\geq1$ and $n\geq 0$ be integers.
We have
\be
\int_{\U(N)} |Z_A^{(n)}(1)|^{2M} |Z_A (1)|^{2k-2M}  dA_N = 
b_{k,M}(n,0)
N^{k^2+2Mn} +O(N^{k^2+2Mn-1}),
\ee
where 
\beas
b_{k,M}(n,0) & = &
(-1)^{nM+\frac{k(k-1)}{2}}
N^{k^2+2Mn} (n!)^{2M}
\sum_{ \substack{s_{l,q}\geq 0\\ \sum_{l=1}^k s_{l,q} \leq n \\
q=1,\ldots,2M}
} (-1/2)^{2Mn-\sum_{q=1}^{2M}\sum_{l=1}^k s_{l,q}} \\
&& 
\hspace{-1.4cm}
\times \, \left( \prod_{q=1}^{2M}\frac{1}{(n-\sum_{l=1}^{k}s_{l,q})!} \right) \left(\prod_{i=1}^k \frac{1}{(2k-i+\sum_{q=1}^{2M} s_{i,q})!} \right) \prod_{1\leq i<j\leq k} \left( \sum_{q=1}^{2M} s_{j,q}- \sum_{q=1}^{2M}s_{i,q}-j+i\right).
\eeas
\end{thm}

Following the philosophy of random matrix theory \cite{keating2000random1, gonek2007hybrid,conrey2005integral}, we make the following conjecture for the moments of higher order derivatives of the Riemann zeta function $\zeta(s)$, on the critical line ($\rm{Re}(s)=1/2$), and of Hardy's $Z$-function defined as
\be
\label{Hardy Z-function}
Z(t) = \frac{\pi^{-\i t/2} \Gamma(1/4 + \i t/2)}{|\Gamma(1/4+\i t/2)|} \zeta(1/2+\i t).
\ee

\begin{con}\label{conjecture for higher derivatives}
For any integers $n_1,n_2 \geq 0$ and integers $k\geq1, 0 \leq M \leq k$,
\be
\frac{1}{T} \int_0^T |\zeta^{(n_1)}(1/2+\i t)|^{2M} |\zeta^{(n_2)}(1/2+\i t)|^{2k-2M}  dt \sim a_{k,M}(n_1,n_2) c_k (\log T)^{k^2+2Mn_1+2(k-M)n_2},
\label{conjecture1}
\ee
and
\be
\frac{1}{T} \int_0^T |Z^{(n_1)}(t)|^{2M} |Z^{(n_2)}(t)|^{2k-2M} dt \sim b_{k,M}(n_1,n_2) c_k (\log T)^{k^2+2Mn_1+2(k-M)n_2},
\label{conjecture2}
\ee
where $c_k$ is the arithmetic factor
\[
c_k = \prod_{p} \left(1 - \frac{1}{p} \right)^{k^2} \sum_{m=0}^\infty \left( \frac{\Gamma(m+k)}{m! \cdot \Gamma(k)} \right)^2  p^{-m}.
\]
\end{con}

It follows from 
\cite[Theorem A$''$]{ingham1928mean} that the left hand side of (\ref{conjecture1}) is asymptotically close to $\frac{1}{2n+1} (\log T)^{2n+1}$ when $k=M=1, n_1=n, n_2=0$. It follows from \cite[Theorem 3]{hall1999behaviour} that the left hand side of (\ref{conjecture2}) is asymptotically close to $\frac{1}{4^n(2n+1)} (\log T)^{2n+1}$ when $k=M=1, n_1=n, n_2=0$. In addition, by \cite[Theorems 10 and 11]{hall1999behaviour}, the left hand side of (\ref{conjecture2}) is asymptotically close to $\frac{1}{672 \pi^2} (\log T)^{8}$ when $k=2, M=1, n_1=2, n_2=0$, and asymptotically close to $\frac{19}{604800 \pi^2} (\log T)^{10}$ when $k=2, M=1, n_1=2, n_2=1$. Moreover, from \cite[Lemma 1]{hall2004stationary}, the left hand side of (\ref{conjecture2}) is asymptotically close to $\frac{17}{1774080\pi^2}(\log T)^{12}$ when $k=M=2$, $n_{1}=2,n_{2}=0$. All these results verify the conjecture in the respective special cases due to the fact that
$c_1=1, c_2=6/\pi^2$, $b_{2,1}(2,0)=1/4032,
b_{2,1}(2,1)=19/3628800$, $b_{2,2}(2,0)=17/10644480$, and the following results on $a_{1,1}(n,0)$ and $b_{1,1}(n,0)$ from Theorems \ref{intro:thm3} and \ref{intro:thm4}.

\begin{prop}
\label{intro:prop3}
For any integer $n\geq 0$,
\[
a_{1,1}(n,0) = 
\frac{1}{2n+1},
\quad 
b_{1,1}(n,0) = \frac{1}{(2n+1)4^{n}}.
\]
\end{prop}

Conjecture \ref{conjecture for higher derivatives} is closely related to problems concerning large gaps between zeros of higher order derivatives of the Hardy $Z$-function. Let
\[
\Lambda^{(m)}= \limsup_{n\rightarrow \infty} \frac{t_{n+1}^{(m)} - t_n^{(m)}}{2\pi/\log t_n^{(m)} },
\]
where $\{t_n^{(m)}\}$ is the sequence of non-negative zero of $Z^{(m)}(t)$, counted according to multiplicity. (This is standard notation: $\Lambda^{(m)}$ is not to be confused with the $m$-th derivative of $\Lambda_A(e^{\i \theta})$.)  It is noteworthy that for any fixed $m$, the quantity $\frac{t_{n+1}^{(m)} - t_n^{(m)}}{2\pi/\log t_n^{(m)} }$ has average value 1 under the Riemann hypothesis. This is well known when $m=0$. For $m\geq 1$, we refer readers to \cite{speiser1935geometrisches,MR1487979,MR417074,MR1354047,MR1084176,MR266874,Zhang} for distribution of zeros of higher order derivatives of the Riemann zeta function and \cite{hall2004stationary} for zeros of derivatives of $Z(t)$.

For $m=0$, the quantity $\Lambda^{(0)}$ was introduced by Hall \cite{hall1999behaviour}.
A related quantity $\lambda$ may be defined with respect to $\{\gamma_n\}$, where $\{\gamma_n\}$ is the non-decreasing sequence of imaginary parts of the zeros of the Riemann zeta function in the upper half plane, counted by multiplicity. Selberg showed that $\lambda>1$ \cite{selberg1946zeta}. Conrey, Ghosh and Gonek showed that $\lambda>2.337$ under the Riemann hypothesis \cite{conrey1984note} and $\lambda>2.68$ under the generalized Riemann hypothesis \cite{conrey1986large}. Note that  $\lambda= \Lambda^{(0)}$ under these assumptions. In \cite{hall2004large}, Hall proposed a method, which does not assume the Riemann hypothesis, to use joint moments of $Z(t)$ and $Z'(t)$ to obtain lower bounds of $\Lambda^{(0)}$. He proved that $\Lambda^{(0)}>4$ times the average gap length assuming (\ref{conjecture2}) for the case $n_1=1,n_2=0, k=6$ in our notation, which was predicted previously by Random Matrix Theory. It is widely believed that $\Lambda^{(0)}=\infty$. This problem remains open. For $m\geq 1$, we refer to \cite{conrey1985mean,hall2004stationary} for lower bounds on $\Lambda^{(1)}$ (which all exceed 1). Additionally, in \cite{bui2023note} the authors also obtained lower bounds of $\Lambda^{(m)}$ for some $m\geq 2$.

We remark that the method in \cite{hall2004large} uses a Wirtinger-type inequality, which involves joint moments of $Z(t)$ and $Z'(t)$. It was mentioned in \cite{bui2023note} that this inequality may be generalized to higher order derivatives that involve joint moments of higher order derivatives of $Z(t)$. Using the results in this paper, we can give numerical data on $b_{k,M}(n_1,n_2)$ for large $k,M$ and $n_1,n_2$. Under Conjecture \ref{conjecture for higher derivatives}, we are able to predict the joint moments of the $n_1$-th and $n_2$-th derivatives of $Z(t)$ for large $n_1,n_2,k$ and for any $M$ between 0 and $k$. It is possible to use these numerical values to predict sharper lower bounds on $\Lambda^{(m)}$ through generalized Wirtinger-type inequalities. We plan to explore this in future research.

\subsection{A brief summary of our main methods}

Our starting point is \cite{conrey2006moments}, which concerns the first order derivative of moments of characteristic polynomials of CUE matrices. The approach there is to take the first order derivative of moments of characteristic polynomials with shifts. In the higher-order case, it is natural to use the same approach and take the $n$-th derivative. The difficulty here is how to represent multiple integrals (see formulae (\ref{eq3}) and (\ref{519eq4})) by expressions in terms of the determinant involving modified Bessel functions of the first kind, which is closely related to solutions of the $\sigma$-Painlev\'{e} III$'$ equation. Furthermore, it is desirable that this expression is computable. Our main idea here is to introduce $n$ extra variables $t_1,\ldots,t_n$, and to compute higher order derivatives with respect to them at 0. We find that the higher order derivative with respect to $t_1$ is related to taking the derivative of a determinant of a matrix whose entries involve the modified Bessel functions of the first kind, and higher order derivatives with respect to the other $n-1$ variables contribute combinatorial coefficients and shifts to the indices of the modified Bessel functions of the first kind, see Propositions \ref{prop1}, \ref{prop18}.

Our method to describe $a_{k,M}(n_1,n_2), b_{k,M}(n_1,n_2)$ is different from that of \cite[Theorem 3]{conrey2006moments}. The main tool of \cite{conrey2006moments} is the Laplacian transform, which converts the complex integral into a computable multiple real integral. In the higher order derivative case, this multiple integral is very complicated. Therefore, instead of using the Laplacian transform, we directly take the higher order derivatives of the determinant involving the modified Bessel functions of the first kind with shifts that appear in (\ref{intro:thm1-eq}). We find that this procedure can be carried out straightforwardly.

In Theorems \ref{intro:thm3} and \ref{intro:thm4}, we provide alternate formulae for the joint moments of higher order derivatives that are more concise than those in Theorems \ref{intro:prop1} and \ref{into:prop2}. Our starting point is similar to that of Theorems \ref{intro:thm1} and \ref{into:thm2}. The difference is that we use a new method (see Proposition \ref{prelimiary:prop1}) to handle the higher order powers of some expressions produced in taking the $n$-th order derivatives.

We remark that our approaches can be generalised to any type of joint moments of characteristic polynomials of the form
\be \label{general case of 1}
\int_{\U(N)} |\Lambda_A^{(n_{1})} (1)|^{2M_1} |\Lambda_A^{(n_{2})}(1)|^{2M_2} \cdots |\Lambda_A^{(n_{d})} (1)|^{2M_d} dA_N.
\ee
and to any corresponding joint moments of the analogue of Hardy's $Z$-functions.

\subsection{Notation}

We denote by $I_n(x)$ the modified Bessel function of the first kind, which has the following power series expansion
\[
I_{n}(x)=(\frac{x}{2})^{n}\sum_{j=0}^{\infty}\frac{x^{2j}}{2^{2j}(n+j)!j!}.
\]
For $x \in \mathbb{C}$, we denote
\bea\label{defofz}
z(x)=\frac{1}{1-e^{-x}}
\eea
The multinomial coefficient is defined as
\[
\binom{m}{l_1,\ldots,l_k} = \frac{m!}{l_1!\cdots l_k!},
\]
where $l_1+\cdots+l_k=m$.
We denote by $S_k$ the set of all permutations of $\{1,\ldots,k\}$. For $n$ an integer, when we say
$h_1+\cdots+h_k=n$ or $h_1+\cdots+h_k\leq n$, we shall always mean $h_i\geq 0$ and taking integer values.
For convenience, for any integrable function $F$, we define
\[
\int_{|w_i|=1} F(w_1,\ldots,w_k) \prod_{i=1}^k dw_i
:=
\int_{|w_1|=1} \cdots \int_{|w_k|=1} F(w_1,\ldots,w_k) \prod_{i=1}^k dw_i.
\]
Let $f(t)$ be a differentiable function.  The $n$-th derivative of $f$ will be denoted as
$(\frac{d}{dt})^n f$ or $f^{(n)}$.
For $w_1,\ldots,w_k\in \mathbb{C}$, the Vandermonde determinant is denoted as
\[
\Delta(w_1,\ldots,w_k) = \det_{k\times k}(w_i^{j-1}) = \prod_{1\leq i<j \leq k} (w_i-w_j).
\]
For convenience, we will ignore the subscripts and denote it as $\Delta(w)$. We also allow commutative differential operators as the arguments, such as
\[
\Delta(\frac{d}{dL}) := \Delta(\frac{d}{dL_1}, \cdots, \frac{d}{dL_n}) = \prod_{1\leq i<j\leq k}\Big(\frac{d}{dL_{j}}-\frac{d}{dL_{i}}\Big).
\]
The square of $\Delta(w)$ will be denoted as $\Delta^2(w)$. Similarly, the composition of the operator $\Delta(\frac{d}{dL})$ will be denoted as $\Delta^2(\frac{d}{dL})$.

\subsection{Acknowledgement}

This research was carried out while F.W. was visiting the University of Oxford. F.W. would like to thank the Mathematical Institute, University of Oxford for generous hospitality. We are grateful to Dr. Yacine Barhoumi-Andr\'{e}ani for helpful feedback on an earlier version of our manuscript. F.W. was supported by the fellowship of China Postdoctoral Science Foundation 2020M670273 and by Jinxin Xue's grant NSFC (project No. 20221300605).  JPK is pleased to acknowledge support from ERC Advanced Grant 740900 (LogCorRM).

\section{Lemmas and propositions}

In this section, we prove some lemmas and propositions that will be used in the proofs of our main results stated in the next section.

The following lemma gives an explicit formula for the higher order derivatives of the function appearing in the expression of moments of characteristic polynomials with shifts.

\begin{lem}
\label{lem1}
Let $n\geq 0, k\geq 1$ be integers, then
\beas
&& \left( \frac{d}{d\alpha} \right)^n \frac{e^{-\frac{N}{2} \alpha}}{\prod_{i=1}^k (w_i-\alpha)} \\
&=& \frac{e^{-\frac{N}{2} \alpha}}{\prod_{i=1}^k (w_i-\alpha)} \sum_{m_1+2m_2+\cdots+nm_n =n} \frac{n!}{m_1!\cdots m_n!} \left( - \frac{N}{2} + \sum_{l=1}^k \frac{1}{w_l-\alpha} \right)^{m_1} \prod_{j=2}^n \left(\frac{1}{j} \sum_{l=1}^k \frac{1}{(w_l-\alpha)^{j}}\right)^{m_j}.
\eeas
\end{lem}

\begin{proof}
It is not hard to check that
\be
\label{eq1}
\left( \frac{d}{d\alpha} \right)^n \frac{e^{-\frac{N}{2} \alpha}}{\prod_{i=1}^k (w_i-\alpha)} = \frac{e^{-\frac{N}{2} \alpha}}{\prod_{i=1}^k (w_i-\alpha)} f_n(\alpha),
\ee
where $f_n(\alpha)$ is defined by the following recursive formulae
\beas
f_1(\alpha) &=& - \frac{N}{2} + \sum_{l=1}^k \frac{1}{w_l-\alpha}, \\
f_{l+1}(\alpha) &=& f_l(\alpha) f_1(\alpha) + f_l'(\alpha) , \quad l\geq 1.
\eeas
Let $g(\alpha)$ be a function such that $g'(\alpha) = f_1(\alpha)$. Consider the function $e^{g(\alpha)}$, then
\bes
\left( \frac{d}{d\alpha} \right)^n e^{g(\alpha)} = e^{g(\alpha)} f_n(\alpha).
\ees
By Fa\`{a} di Bruno's formula
\beas
\left( \frac{d}{d\alpha} \right)^n e^{g(\alpha)}
=e^{g(\alpha)}\sum_{m_1+2m_2+\cdots+nm_n =n} \frac{n!}{m_1!\cdots m_n!} \prod_{j=1}^n \left(\frac{f_1^{(j-1)}(\alpha)}{j!}\right)^{m_j},
\eeas
so
\beas
f_{n}(\alpha)=\sum_{m_1+2m_2+\cdots+nm_n =n} \frac{n!}{m_1!\cdots m_n!} \left( - \frac{N}{2} + \sum_{l=1}^k \frac{1}{w_l-\alpha} \right)^{m_1} \prod_{j=2}^n \left(\frac{1}{j} \sum_{l=1}^k \frac{1}{(w_l-\alpha)^{j}} \right)^{m_j}.
\eeas
By (\ref{eq1}), we obtain the claimed result.
\end{proof}

The proposition below gives a formula for the powers of the expression obtained above at $\alpha=0$.

\begin{prop}
\label{prelimiary:prop1}
Let $k\geq1, r \geq 1$, $n\geq 0$ and $m_{1}\geq 0,\ldots,m_{n}\geq 0$ be integers. Let $w_{1},\ldots,w_{k}$ be non-zero complex numbers. Then we have
\bea
&& \Bigg( \sum_{m_1+2m_2+\cdots+nm_n=n} \frac{n!}{m_1!\cdots m_n!} (-N + \sum_{l=1}^k \frac{1}{w_l})^{m_1} \prod_{j=2}^n (\frac{1}{j} \sum_{l=1}^k \frac{1}{w_l^j})^{m_j}\Bigg)^{r} \nonumber \\
&=&\sum_{ \substack{\sum_{l=1}^k s_{lj} \leq n \\ 
j=1,\ldots,r} } \, \prod_{j=1}^r \left((-N)^{n-\sum_{l=1}^k s_{lj}} \binom{n}{\sum_{l=1}^k s_{lj}} (\sum_{l=1}^k s_{lj})! \right)
 \frac{1}{\prod_{l=1}^k w_l^{\sum_{j=1}^r s_{lj}}}.
 \label{eq1128formula1}
\eea
and 
\bea
\Bigg( \sum_{m_1+2m_2+\cdots+nm_n=n} \frac{n!}{m_1!\cdots m_n!} \prod_{j=1}^n (\frac{1}{j} \sum_{l=1}^k \frac{1}{w_l^j})^{m_j}\Bigg)^{r}
=\sum_{ \substack{\sum_{l=1}^k h_{lj} = n \\ 
j=1,\ldots,r} }  \frac{(n!)^r}{\prod_{l=1}^k w_l^{\sum_{j=1}^r h_{lj}}}
\label{eq1128formula2}.
\eea 
\end{prop}
\begin{proof}
Denote
\beas
a_{n-m_1}= \sum_{2m_2+\cdots+nm_n=n-m_1} \frac{1}{m_2!\cdots m_n!} \prod_{j=2}^n (\frac{1}{j} \sum_{l=1}^k \frac{1}{w_l^j})^{m_j}.
\eeas
Then $a_{n-m_1}$ is the coefficient of 
$t^{n-m_1}$ in the Taylor expansion of
$$
\prod_{j=2}^n e^{ \frac{t^j}{j} \sum_{l=1}^k \frac{1}{w_l^j} }
=
\prod_{j=2}^n
\prod_{l=1}^k e^{ \frac{1}{j} \frac{t^j}{w_l^j} }.
$$
So
\beas
a_{n-m_1} &=& \frac{1}{(n-m_1)!} \frac{d^{(n-m_1)}}{dt^{(n-m_1)}} e^{\sum_{j=2}^n (\frac{1}{j} \sum_{l=1}^k \frac{1}{w_l^j} )t^j} \Bigg|_{t=0} \\
&=& \frac{1}{(n-m_1)!} \frac{d^{(n-m_1)}}{dt^{(n-m_1)}} e^{\sum_{j=2}^\infty (\frac{1}{j} \sum_{l=1}^k \frac{1}{w_l^j} )t^j} \Bigg|_{t=0} \\
&=& \frac{1}{(n-m_1)!} \frac{d^{(n-m_1)}}{dt^{(n-m_1)}} \prod_{l=1}^k  e^{\sum_{j=2}^\infty \frac{1}{j} \frac{t^j}{w_l^j} } \Bigg|_{t=0} \\
&=& \frac{1}{(n-m_1)!} \frac{d^{(n-m_1)}}{dt^{(n-m_1)}} \prod_{l=1}^k  e^{-\ln(1-\frac{t}{w_l}) - \frac{t}{w_l} } \Bigg|_{t=0} \\
&=& \frac{1}{(n-m_1)!} \frac{d^{(n-m_1)}}{dt^{(n-m_1)}} \prod_{l=1}^k  e^{- \frac{t}{w_l} } \frac{1}{1-\frac{t}{w_l} }\Bigg|_{t=0}.
\eeas
Namely, $a_{n-m_1}$ is the coefficient of 
$t^{n-m_1}$ in the Taylor expansion of
\[
F(t) =\left(\prod_{l=1}^k  \frac{1}{1-\frac{t}{w_l} }\right)
e^{- \sum_{l=1}^k\frac{t}{w_l} }.
\]

Let $G_1(t)= e^{(-N+\sum_{l=1}^k\frac{1}{w_l}) t}$ and let $G_2(t)= e^{(\sum_{l=1}^k\frac{1}{w_l}) t}$, then
\beas
&&\sum_{m_1+2m_2+\cdots+nm_n=n} \frac{n!}{m_1!\cdots m_n!} (-N + \sum_{l=1}^k \frac{1}{w_l})^{m_1} \prod_{j=2}^n (\frac{1}{j} \sum_{l=1}^k \frac{1}{w_l^j})^{m_j} \nonumber \\
&=& \sum_{m_1=0}^n  a_{n-m_1} \frac{n!}{m_1!} (-N + \sum_{l=1}^k \frac{1}{w_l})^{m_1}  \nonumber \\
&=& \frac{d^n}{dt^n} (G_1(t) F(t))\Bigg|_{t=0} =\frac{d^n}{dt^n} \left(e^{-Nt} \prod_{l=1}^k  \frac{1}{1-\frac{t}{w_l} } \right) \Bigg|_{t=0} \nonumber \\
&=& \frac{d^n}{dt^n} \left( e^{-Nt} \prod_{l=1}^k (\sum_{i=0}^n \frac{t^i}{w_l^i}) \right) \Bigg|_{t=0} .
\eeas
Similarly,
\bes
\sum_{m_1+2m_2+\cdots+nm_n=n} \frac{n!}{m_1!\cdots m_n!} \prod_{j=1}^n (\frac{1}{j} \sum_{l=1}^k \frac{1}{w_l^j})^{m_j} 
=\frac{d^n}{dt^n} \left( \prod_{l=1}^k (\sum_{i=0}^n \frac{t^i}{w_l^i}) \right) \Bigg|_{t=0} .
\ees

Note that
\bea
&& \left(\frac{d^n}{dt^n} \left( e^{-Nt} \prod_{l=1}^k (\sum_{i=0}^n \frac{t^i}{w_l^i}) \right) \Bigg|_{t=0} \right)^r \nonumber \\
&=& \left(\prod_{j=1}^r \frac{d^n}{dt_j^n} \right) \prod_{j=1}^r \left(  e^{-Nt_j} \prod_{l=1}^k (\sum_{i=0}^n \frac{t_j^i}{w_l^i}) \right) \Bigg|_{t_1,\ldots,t_r=0} \nonumber  \\
&=& \left(\prod_{j=1}^r \frac{d^n}{dt_j^n} \right)  e^{-N(t_1+\cdots+t_r)} \prod_{j=1}^r \prod_{l=1}^k (\sum_{i=0}^n \frac{t_j^i}{w_l^i})  \Bigg|_{t_1,\ldots,t_r=0} \nonumber  \\
&=& \left(\prod_{j=1}^r \frac{d^n}{dt_j^n} \right)  e^{-N(t_1+\cdots+t_r)} 
\sum_{ \substack{s_{11},\ldots,s_{1r},  \\
\cdots\cdots\cdots \\
s_{k1},\ldots,s_{kr}= 0} }^n \frac{\prod_{j=1}^r t_j^{\sum_{l=1}^k s_{lj}}}{\prod_{l=1}^k w_l^{\sum_{j=1}^r s_{lj}}}
\Bigg|_{t_1,\ldots,t_r=0} \nonumber \\
&=& \sum_{ \substack{s_{11},\ldots,s_{1r}, \\
\cdots\cdots\cdots \\
s_{k1},\ldots,s_{kr}= 0} }^n \prod_{j=1}^r \left(\frac{d^n}{dt_j^n} e^{-Nt_j}  t_j^{\sum_{l=1}^k s_{lj}} \right)
\Bigg|_{t_1,\ldots,t_r=0}  \frac{1}{\prod_{l=1}^k w_l^{\sum_{j=1}^r s_{lj}}}\nonumber  \\ 
&=& \sum_{ \substack{\sum_{l=1}^k s_{lj} \leq n \\ 
j=1,\ldots,r} } \prod_{j=1}^r \left((-N)^{n-\sum_{l=1}^k s_{lj}} \binom{n}{\sum_{l=1}^k s_{lj}} (\sum_{l=1}^k s_{lj})! \right)
 \frac{1}{\prod_{l=1}^k w_l^{\sum_{j=1}^r s_{lj}}}. \nonumber
\eea
Similarly,
\bes
\left(\frac{d^n}{dt^n} \left( \prod_{l=1}^k (\sum_{i=0}^n \frac{t^i}{w_l^i}) \right) \Bigg|_{t=0} \right)^r 
= \sum_{ \substack{ \sum_{l=1}^k h_{lj} = n \\ 
j=1,\ldots,r} } (n!)^r \frac{1}{\prod_{l=1}^k w_l^{\sum_{j=1}^r h_{lj}}}.
\ees
Putting it all together, we obtain the claimed result.
\end{proof}

\begin{lem}[Equation (2.11) of \cite{conrey2006moments}]
\label{lem2}
Let $n\geq 0$ be an integer, then
\[
\frac{1}{2\pi \i} \int_{|w|=1} \frac{e^{Lw+\frac{t}{w}}}{w^{n+1}} dw
= \left(\frac{L}{t}\right)^{n/2} I_n(2\sqrt{Lt}) .
\]
\end{lem}


In our situation, we need to handle the expression (\ref{lem3:eq}). We give a formula for it correlated with the modified Bessel functions of the first kind.

\begin{lem}
\label{lem3}
For integers $n\geq 1, k \geq 1$, we have
\bea
&& \frac{1}{2\pi \i} \int_{|w|=1} \frac{e^{Lw+ \frac{t_1}{w} + \cdots + \frac{t_n}{w^n} }}{w^{2k}} dw \label{lem3:eq} \\
&=& \sum_{m_2,\ldots,m_n=0}^\infty \, \left(\prod_{j=2}^n \, \frac{t_j^{m_j}}{m_j!} \right) \, I_{\sum_{j=2}^n jm_j + 2k-1}(2\sqrt{Lt_1}) \left(\frac{L}{t_1}\right)^{\frac{\sum_{j=2}^n jm_j + 2k-1}{2}}.
\nonumber 
\eea
\end{lem}

\begin{proof}
We need to compute the coefficient of $w^{2k-1}$ in the Taylor expansion of $e^{Lw+ \frac{t_1}{w} + \cdots + \frac{t_n}{w^n} }$. Let $a_m(i)$ be the coefficient of $w^i$ in the Taylor expansion of $e^{Lw+ \frac{t_1}{w} + \cdots + \frac{t_m}{w^m} }$ for $m=1,2,\ldots,n$, then  $a_n(2k-1) = \sum_{m_n=0}^\infty \frac{t_n^{m_n}}{m_n!} a_{n-1}(nm_n+2k-1)$. So by Lemma \ref{lem2},
\beas
a_n(2k-1) &=& \sum_{m_2,\ldots,m_n=0}^\infty \, \left(\prod_{j=2}^n \, \frac{t_j^{m_j}}{m_j!} \right) a_1 (\sum_{j=2}^n jm_j + 2k-1 ) \\
&=& \sum_{m_2,\ldots,m_n=0}^\infty \, \left(\prod_{j=2}^n \, \frac{t_j^{m_j}}{m_j!} \right) \, I_{\sum_{j=2}^n jm_j + 2k-1}(2\sqrt{Lt_1}) \left(\frac{L}{t_1}\right)^{\frac{\sum_{j=2}^n jm_j + 2k-1}{2}}.
\eeas
\end{proof}

We now state a fact about higher order derivatives of determinants of functions.

\begin{lem}
\label{key-lem3}
Let $s\geq 0$, $k\geq 1$ be integers and $a_{i,j}(x)$ be $s$-th differentiable functions of $x$. Then
\[\Big(\frac{d}{dx}\Big)^{s}\det_{k\times k}(a_{i,j}(x))=\sum_{\substack{l_1+\cdots+l_k=s \\ l_1\geq 0,\ldots,l_k\geq 0}}\binom{s}{l_1,\ldots,l_k}\det_{k\times k}\Big(a_{i,j}^{(l_{i})}(x)\Big),\]
where $a_{i,j}^{(l_{i})}(x)$ means that we take the $l_{i}$-th derivative of $a_{i,j}(x)$.
\end{lem}

\begin{proof}
Expand the determinant as a sum over all permutations $\sigma$ of the numbers $1,2,\ldots,k$:
\[\det_{k\times k}(a_{i,j}(x))=\sum_{\sigma \in S_k}\text{sign}(\sigma)\prod_{i=1}^{k}a_{i,\sigma(i)}(x).\]
Then
\beas
\Big(\frac{d}{dx}\Big)^{s}\det_{k\times k}(a_{i,j}(x)) &=& \sum_{\sigma \in S_k}\text{sign}(\sigma) \Big(\frac{d}{dx}\Big)^{s}\Big(\prod_{i=1}^{k}a_{i,\sigma(i)}(x)\Big) \\
&=&\sum_{\sigma \in S_k}\text{sign}(\sigma) \sum_{\substack{l_1+\cdots+l_k=s \\ l_1\geq 0,\ldots,l_k\geq 0}}\binom{s}{l_1,\ldots,l_k}\prod_{i=1}^{k}(\frac{d}{dx})^{l_{i}}(a_{i,\sigma(i)}(x))\\
&=&\sum_{\substack{l_1+\cdots+l_k=s \\ l_1\geq 0,\ldots,l_k\geq 0}}\binom{s}{l_1,\ldots,l_k}\sum_{\sigma \in S_k}\text{sign}(\sigma) \prod_{i=1}^{k}(\frac{d}{dx})^{l_{i}}(a_{i,\sigma(i)}(x))\\
&=&\sum_{\substack{l_1+\cdots+l_k=s \\ l_1\geq 0,\ldots,l_k\geq 0}}\binom{s}{l_1,\ldots,l_k}\det_{k\times k}\Big((\frac{d}{dx})^{l_{i}}(a_{i,j}(x))\Big).
\eeas
\end{proof}

The following are two identities concerning $\Lambda_{A}^{(n)}(1)$ and $Z_{A}^{(n)}(1)$. They ensure that the derivatives with respect to the shifts at 0 give $|\Lambda_{A}^{(n)}(1)|^2$ and $|Z_{A}^{(n)}(1)|^2$.

\begin{lem}\label{15}
For $n\geq 0$, 
\[\Lambda_{A^*}^{(n)}(1)=\overline{\Lambda_{A}^{(n)}(1)}\]
and
\[Z_{A^*}^{(n)}(1)=(-1)^{N} \overline{Z_{A}^{(n)}(1)}.\]
\end{lem}
\begin{proof}
By the definition of $\Lambda_{A}(s)$, it is not hard to check that $\Lambda_{A^*}^{(n)}(1)=\overline{\Lambda_{A}^{(n)}(1)}$.
Since
\beas
Z_{A}(s) &=& e^{-\pi \i N/2}e^{\i \sum_{n=1}^{N}\theta_{n}/2}s^{-       N/2}\Lambda_{A}(s), \\
Z_{A^{*}}(s) &=& (-1)^{N}e^{\pi \i N/2}e^{-\i \sum_{n=1}^{N}\theta_{n}/2}s^{-       N/2}\Lambda_{A^{*}}(s).
\eeas
Combining with $\Lambda_{A^*}^{(n)}(1)=\overline{\Lambda_{A}^{(n)}(1)}$, we have
$Z_{A^*}^{(n)}(1)=(-1)^{N} \overline{Z_{A}^{(n)}(1)}$.
\end{proof}

We use $\Xi$ to denote the subset of permutations $\sigma\in S_{2k}$ of $\{1,2,\ldots,2k\}$ for which $\sigma(1)<\sigma(2)<\cdot\cdot\cdot<\sigma(k)$ and $\sigma(k+1)<\sigma(k+2)<\cdot\cdot\cdot<\sigma(2k)$.
The following lemma was proved in section 2 of \cite{conrey2003autocorrelation}.

\begin{lem}
\label{11101}
Assume that $\alpha_{1},\ldots,\alpha_{k}$ are distinct complex numbers. We have
\[\int_{\U(N)}\prod_{j=1}^{k}\Lambda_{A}(e^{-\alpha_{j}})\Lambda_{A^{*}}(e^{\alpha_{j+k}})dA_{N}=\sum_{\sigma\in \Xi}e^{N\sum_{j=1}^{k}(\alpha_{\sigma(j)}-\alpha_{j})}\prod_{1\leq i,j\leq k}z(\alpha_{\sigma(j)}-\alpha_{\sigma(k+i)}),\]
where $z(x)$ is defined in (\ref{defofz}).
\end{lem}
By the above lemma and the definition of $Z_{A}(s)$,
\beas
&& \int_{\U(N)} \prod_{j=1}^k Z_A(e^{-\alpha_j}) Z_{A^*}(e^{\alpha_{j+k}})  dA_N \\
&=& (-1)^{Nk} e^{-\frac{N}{2} \sum_{j=1}^{2k} \alpha_j} \sum_{\sigma \in\Xi} e^{N\sum_{j=1}^k \alpha_{\sigma(j)}} \prod_{1\leq i,j\leq k} Z(\alpha_{\sigma(j)} - \alpha_{\sigma(k+i)})
\eeas
By the residue theorem, the summations in Lemma \ref{11101} and the above formula can be expressed as integrals in the following way.

\begin{lem}[Lemma 3 of \cite{conrey2006moments}]
\label{11/10lemma1}
\beas
&& \int_{\U(N)} \prod_{j=1}^k \Lambda_A(e^{-\alpha_j}) \Lambda_{A^*}(e^{\alpha_{j+k}})  dA_N \\
&=&\frac{1}{k!(2\pi \i)^{k}}
\int_{|w_{i}|=1}e^{N\sum_{j=1}^{k}(w_{j}-\alpha_{j})}\prod_{\substack{1\leq i\leq k\\ 1\leq j\leq 2k}}z(w_{i}-\alpha_{j})\prod_{i\neq j}z(w_{i}-w_{j})^{-1}\prod_{j=1}^{k}dw_{j},
\eeas
and
\beas
&& \int_{\U(N)} \prod_{j=1}^k Z_A(e^{-\alpha_j}) Z_{A^*}(e^{\alpha_{j+k}})  dA_N \\
&=&(-1)^{Nk}\frac{e^{-\frac{N}{2}\sum_{j=1}^{2k}\alpha_{j}}}{k!(2\pi \i)^{k}} \int_{|w_{i}|=1}
e^{N\sum_{j=1}^{k}w_{j}}\prod_{\substack{1\leq i\leq k\\1\leq j\leq2k}}z(w_{i}-\alpha_{j})\prod_{i\neq j}z(w_{i}-w_{j})^{-1}\prod_{j=1}^{k}dw_{j}.
\eeas
\end{lem}

\begin{rem}\label{remark1}
\rm{ The integral contours $\{w_{i}:|w_{i}|=1\}$ can be replaced by any contours enclosing the variables $\alpha_{i}$, by the residue theorem.
}
\end{rem}

The following lemma is a non-symmetric version of \cite[Lemma 5]{conrey2006moments}.

\begin{lem} \label{lem2211}
Let $f_1,\ldots,f_k$ be $2k-2$ times differentiable functions. Let $\Delta(\frac{d}{dL})$ be the differential operator
\[
\Delta(\frac{d}{dL}):=\prod_{1\leq i<j\leq k}\Big(\frac{d}{dL_{j}}-\frac{d}{dL_{i}}\Big)=\det_{k\times k}\Big(\frac{d^{j-1}}{dL_{i}^{j-1}}\Big).
\]
Then
\[
\Delta^2(\frac{d}{dL}) \prod_{i=1}^k f_i(L_i)  \Bigg|_{L_i=X} \\
= \sum_{\sigma \in S_k} \det_{k\times k} (f_{\sigma(i)}^{(i + j-2)}(X) ) .
\]
In particular, when $f_1=\cdots=f_k=f$, we have
\[
\Delta^2(\frac{d}{dL})(\prod_{i=1}^{k}f(L_{i}))\Big|_{L_i=X}=k!\det_{k\times k}(f^{(i+j-2)}(X)) .
\]
\end{lem}

\begin{proof}
Firstly,
\beas
\Delta(\frac{d}{dL}) \prod_{i=1}^k f_i(L_i) 
&=& \det_{k\times k} \left( \frac{d^{j-1}}{dL_i^{j-1}} \right) \prod_{i=1}^k f_i(L_i) \\
&=& \sum_{\sigma\in S_k} {\rm sign}(\sigma) 
\left(\frac{d}{dL_1}  \right)^{\sigma(1)-1} \left(\frac{d}{dL_2}  \right)^{\sigma(2)-1}\cdots \left(\frac{d}{dL_k}  \right)^{\sigma(k)-1} \prod_{i=1}^k f_i(L_i) \\
&=& \sum_{\sigma \in S_k} {\rm sign}(\sigma) \prod_{i=1}^k f_i^{(\sigma(i)-1)}(L_i)  \\
&=& \det_{k\times k} (f_i^{(j-1)}(L_i)).
\eeas
Secondly,
\beas
\Delta^2(\frac{d}{dL}) \prod_{i=1}^k f_i(L_i)  \Bigg|_{L_i=X} 
&=& \Delta(\frac{d}{dL}) \det_{k\times k} (f_i^{(j-1)}(L_i))  \Bigg|_{L_i=X} \\
&=& \Delta(\frac{d}{dL}) \sum_{\sigma} {\rm sign}(\sigma) \prod_{i=1}^k f_i^{(\sigma(i)-1)}  (L_i)\Bigg|_{L_i=X} \\
&=& \sum_{\tau\in S_k} \sum_{\sigma\in S_k} {\rm sign}(\sigma) 
 {\rm sign}(\tau) \prod_{i=1}^k f_i^{(\sigma(i)-1+\tau(i)-1)}  (L_i)\Bigg|_{L_i=X} \\
&=&  \sum_{\sigma\in S_k} {\rm sign}(\sigma) 
\det_{k\times k} (f_i^{(\sigma(i) + j-2)}(X) ) \\
&=& \sum_{\sigma\in S_k} \det_{k\times k} (f_{\sigma(i)}^{(i + j-2)}(X) ) .
\eeas
\end{proof}

The following proposition will be used in the proofs of our main theorems \ref{general mixed moment on Lambda} and \ref{general 1128mixed moment}.

\begin{prop}
\label{lemmaonthe integral}

Let $k\geq 1$, $Q_{l}\geq 0$ be integers for $l=1,\dots,k$. Then we have
\beas
\frac{1}{(2\pi \i)^{k}} \int_{|w_i|=1}  \frac{e^{\sum_{j=1}^k N w_j}\Delta^2(w)}{\prod_{l=1}^k w_l^{2k+Q_{l}}} \prod_{j=1}^k dw_j 
= \sum_{\sigma \in S_k} \det_{k\times k} \left( \frac{N^{2k+Q_{\sigma(i)}-1-(i+j-2)}}{(2k+Q_{\sigma(i)}-1-(i+j-2))!} \right).
\eeas

\end{prop}

\begin{proof}
Note that
\bea
\frac{1}{(2\pi \i)^{k}}\int_{|w_i|=1} \frac{e^{\sum_{j=1}^k N w_j}\Delta^2(w)}{\prod_{l=1}^k w_l^{2k+Q_{l}}} \prod_{j=1}^k dw_j 
=
\frac{1}{(2\pi \i)^{k}} \Delta^2(\frac{d}{dL}) \int_{|w_i|=1} \frac{e^{\sum_{j=1}^k w_jL_j}}{\prod_{l=1}^k w_l^{2k+Q_{l}}} \prod_{j=1}^k dw_j \Bigg|_{L_j=N}.
\label{eq6}
\eea
Let
\[
f_j(L_j) := \frac{1}{2\pi \i}  \int_{|w|=1} \frac{e^{L_jw}}{w^{2k+Q_{j}}} dw.
\]
Then
\[(\ref{eq6})
=\Delta^2(\frac{d}{dL})  \prod_{j=1}^k f_j(L_j) \Bigg|_{L_j=N}.
\]
By Lemma \ref{lem2211},
\beas
(\ref{eq6}) = \sum_{\sigma \in S_k} \det_{k\times k} (f_{\sigma(i)}^{(i+j-2)} (N)).
\eeas

Note that $f_j(L_j)$ is the coefficient of $w^{2k+Q_{j}-1}$ in the Taylor expansion of $e^{L_j w}$, namely
\[
f_i(L_i) = \frac{(L_i)^{2k+Q_{i}-1}}{(2k+Q_{i}-1)!}.
\]
So
\beas
(\ref{eq6})
= \sum_{\sigma\in S_k} \det_{k\times k} \left( \frac{N^{2k+ Q_{\sigma(i)}-1-(i+j-2)}}{(2k+Q_{\sigma(i)}-1-(i+j-2))!} \right).
\eeas
\end{proof}

The following lemma follows from the proof of \cite[Theorem 3, see (4.39)-(4.41)]{conrey2006moments}.

\begin{lem}\label{determinant}
\[\det_{k\times k}\Big(\frac{1}{(2k+1+m_{i}-i-j)!}\Big)=\prod_{i=1}^{k}\frac{1}{(2k-i+m_{i})!}\prod_{1\leq i<j\leq k}(m_{j}-m_{i}-j+i).\]
\end{lem}

\section{Main results}

In the following, we first calculate the integrals
\bea
\label{eq3}
\frac{1}{(2\pi \i)^k k!}   \int_{|w_i|=1} \frac{e^{N \sum_{i=1}^k w_i}}{\prod_{i=1}^k w_i^{2k}} \Delta^2(w) \left(\sum_{l=1}^k \frac{1}{w_l}  - \frac{N}{2} \right)^{m_1} \prod_{j=2}^n \left(\sum_{l=1}^k \frac{1}{w_l^j} \right)^{m_j} \prod_{i=1}^k dw_i
\eea
and
\be
\label{519eq4}
\frac{1}{(2\pi \i)^k k!}   \int_{|w_i|=1} \frac{e^{N \sum_{i=1}^k w_i}}{\prod_{i=1}^k w_i^{2k}} \Delta^2(w) \left(\sum_{l=1}^k \frac{1}{w_l}  - N\right)^{m_1} \prod_{j=2}^n \left(\sum_{l=1}^k \frac{1}{w_l^j} \right)^{m_j} \prod_{i=1}^k dw_i
\ee

\begin{prop}
\label{prop1}
Let $(m_1,\ldots,m_n)$ be a tuple of non-negative integers, then
\beas
(\ref{eq3}) &=& N^{k^2+\sum_{s=1}^n sm_s} \sum_{\substack{
\sum_{l=1}^k m_{sl} = m_s \\ 
s=2,\ldots,n}} 
\left(\,\prod_{s=2}^n \binom{m_s}{m_{s1},\ldots,m_{sk}} \right)
\frac{d^{m_1}}{dx^{m_1}} \Bigg( e^{-\frac{x}{2}} x^{-\frac{k^2}{2}-\frac{1}{2} \sum_{s=2}^n sm_s} \\
&&  \times \, \det_{k\times k} \left( 
I_{\sum_{s=2}^n sm_{si} + i+j-1}(2\sqrt{x})  \right)\Bigg) \Bigg|_{x=0}.
\eeas
More explicitly,
\beas
(\ref{eq3}) &=&  N^{k^2+\sum_{s=1}^n sm_s} \sum_{\substack{
\sum_{l=1}^k m_{sl} = m_s \\ s=2,\ldots,n \\
\sum_{l=0}^k m_{1l}=m_1}} 
 (-\frac{1}{2})^{m_{10}}
\left(\,\prod_{s=2}^n \binom{m_s}{m_{s1},\ldots,m_{sk}} \right) \binom{m_1}{m_{10},\ldots, m_{1k}}  \\
&&  \times \, \left(\prod_{i=1}^k \frac{1}{(2k+\sum_{s=1}^n sm_{si}-i)!} \right)
\prod_{1\leq i<j\leq k} \left(\sum_{s=1}^n sm_{sj} - \sum_{s=1}^n sm_{si}-j+i\right).
\eeas
\end{prop}
\begin{proof}
In the following, we write $t=(t_1,\ldots,t_n)$ and when we say $t=0$ we mean $t_1=\cdots=t_n=0$.
We rewrite
\beas
(\ref{eq3}) = \frac{1}{(2\pi \i)^k k!} \prod_{i=1}^n \left( \frac{d}{dt_i} \right)^{m_i} e^{-\frac{N}{2}t_1} \Delta^2(\frac{d}{dL})   \int_{|w_i|=1}  \frac{e^{\sum_{i=1}^k(L_iw_i+\frac{t_1}{w_i}+ \frac{t_2}{w_i^2}+\cdots+\frac{t_n}{w_i^n}) }}{\prod_{i=1}^k 
 w_i^{2k}} \prod_{i=1}^k dw_i \Bigg|_{L_i=N,t=0}.
\eeas
Let
\[
f(x) = \frac{1}{2\pi \i} \int_{|w|=1} \frac{e^{xw+\sum_{j=1}^n\frac{t_j}{w^j}}}{w^{2k}} dw.
\]
By Lemma \ref{lem2211}, we have
\beas
(\ref{eq3}) = \prod_{i=1}^n \left( \frac{d}{dt_i} \right)^{m_i} e^{-\frac{N}{2}t_1} \det_{k\times k} (f^{(i+j-2)}(N))\Bigg|_{t=0}.
\eeas
By Lemma \ref{lem3},
\[
f(L) = \sum_{l_2,\ldots,l_n=0}^\infty \, \left(\prod_{j=2}^n \, \frac{t_j^{l_j}}{l_j!} \right) \, \sum_{l_1=0}^\infty \frac{L^{l_1} t_{1}^{l_1} L^{\sum_{j=2}^n jl_j+2k-1}}{l_{1}! (\sum_{j=2}^n jl_j+2k-1+l_1)!}.
\]
Since $i+j-2\leq 2k-2$, we have
\[
f^{(i+j-2)}(N) = \sum_{l_2,\ldots,l_n=0}^\infty \, \left(\prod_{j=2}^n \, \frac{t_j^{l_j}}{l_j!} \right) \, \sum_{l_1=0}^\infty \frac{N^{l_1} t_{1}^{l_1} N^{\sum_{j=2}^n jl_j+2k-1-(i+j-2)}}{l_{1}! (\sum_{j=2}^n jl_j+2k-1+l_1-(i+j-2))!}.
\]
Since $\det(a_{i,j}) = \det(a_{k+1-i,k+1-j})$,
\beas
(\ref{eq3}) = \prod_{l=1}^n \left( \frac{d}{dt_l} \right)^{m_l} e^{-\frac{N}{2}t_1}
\det_{k\times k} \left( 
\sum_{l_2,\ldots,l_n=0}^\infty \, \left(\prod_{s=2}^n \, \frac{t_s^{l_s}}{l_s!} \right) I_{\sum_{s=2}^n sl_s + i+j-1}(2\sqrt{Nt_1}) \left(\frac{N}{t_1}\right)^{\frac{\sum_{s=2}^n s l_s + i+j-1}{2}} \right)\Bigg|_{t=0}.
\eeas
By Lemma \ref{key-lem3}, we further obtain
\beas
(\ref{eq3}) = && \sum_{\substack{ m_{sl} \geq 0 \\
\sum_{l=1}^k m_{sl} = m_s \\ 
s=2,\ldots,n}} 
\left( \prod_{s=2}^n \binom{m_s}{m_{s1},\ldots,m_{sk}} \right)
\frac{d^{m_1}}{dt_1^{m_1}} \Bigg( e^{-\frac{N}{2}t_1}\\
&& \times \, \det_{k\times k} \left( 
I_{\sum_{s=2}^n sm_{si} + i+j-1}(2\sqrt{Nt_1}) \left(\frac{N}{t_1}\right)^{\frac{\sum_{s=2}^n s m_{si} +i+j-1}{2}} \right) \Bigg) \Bigg|_{t_1=0}.
\eeas
Let $Nt_1=x$, then
\beas
(\ref{eq3}) &=& N^{k^2+\sum_{s=1}^n sm_s} \sum_{\substack{m_{sl} \geq 0 \\
\sum_{l=1}^k m_{sl} = m_s \\ 
s=2,\ldots,n}} \left( \prod_{s=2}^n \binom{m_s}{m_{s1},\ldots,m_{sk}} \right)
\frac{d^{m_1}}{dx^{m_1}} \Bigg( e^{-\frac{x}{2}} x^{-\frac{k^2}{2}-\frac{1}{2} \sum_{s=2}^n sm_s} \\
&& \times \, \det_{k\times k} \left( 
I_{\sum_{s=2}^n sm_{si} + i+j-1}(2\sqrt{x})  \right) \Bigg) \Bigg|_{x=0}.
\eeas

Concerning the explicit formula of (\ref{eq3}), note that
\beas
&& \frac{d^{m_1}}{dx^{m_1}} \Bigg( e^{-\frac{x}{2}} x^{-\frac{k^2}{2}-\frac{1}{2} \sum_{s=2}^n sm_s} \det_{k\times k} \left( 
I_{\sum_{s=2}^n sm_{si} + i+j-1}(2\sqrt{x})  \right)\Bigg) \Bigg|_{x=0} \\
&=& \hspace{-.3cm} \sum_{\substack{m_{10},\ldots, m_{1k}\geq 0 \\ m_{10}+\ldots+ m_{1k}=m_1 }} \binom{m_1}{m_{10},\ldots, m_{1k}} (-\frac{1}{2})^{m_{10}} \det_{k\times k} \left(\left(  \sum_{l=0}^\infty \frac{x^l}{l! (\sum_{s=2}^n sm_{si}+i+j-1+l)! }\right)^{(m_{1i})}\right)\Bigg|_{x=0} \\
&=&  \hspace{-.3cm} \sum_{\substack{m_{10},\ldots, m_{1k}\geq 0 \\ m_{10}+\ldots+ m_{1k}=m_1}} \binom{m_1}{m_{10},\ldots, m_{1k}} (-\frac{1}{2})^{m_{10}} \det_{k\times k} \left(
\frac{1}{(\sum_{s=1}^n sm_{si}+i+j-1)!}
\right) .
\eeas
We next use Lemma \ref{determinant} to finish this proof. To use Lemma \ref{determinant}, we first change the variables $i\rightarrow k+1-i$, $j\rightarrow k+1-j$ so $m_{si}$ becomes $\tilde{m}_{si}:=m_{s(k+1-i)}$. The summation above over $m_{si}$ has a symmetric property in the sense that $\sum_{i=1}^k m_{si} = \sum_{i=1}^k \tilde{m}_{si}$. By this observation and Lemma \ref{determinant}, we obtain
the claim in this proposition.
\end{proof}

A similar argument to that of Proposition \ref{prop1} leads to the following result for (\ref{519eq4}).

\begin{prop}
\label{prop18}
For any tuple $(m_1,\ldots,m_n)$ of non-negative integers, we have
\beas
(\ref{519eq4}) &=& N^{k^2+\sum_{s=1}^n sm_s} \sum_{\substack{m_{sl}\geq 0 \\
\sum_{l=1}^k m_{sl} = m_s \\ 
s=2,\ldots,n}} \prod_{s=2}^n \binom{m_s}{m_{s1},\ldots,m_{sk}} \left(\frac{d}{dx}\right)^{m_1} \Bigg( e^{-x}x^{-\frac{k^2}{2}-\frac{1}{2} \sum_{s=2}^n sm_s} \\
&& \times \, \det_{k\times k} \left( 
I_{\sum_{s=2}^n sm_{si} + i+j-1}(2\sqrt{x})  \right)\Bigg) \Bigg|_{x=0}.
\eeas
More explicitly,
\beas
(\ref{519eq4}) &=&  N^{k^2+\sum_{s=1}^n sm_s} \sum_{\substack{s=2,\ldots,n\\
\sum_{l=1}^k m_{sl} = m_s \\
\sum_{l=0}^k m_{1l}=m_1}} 
 (-1)^{m_{10}}
\left(\,\prod_{s=2}^n \binom{m_s}{m_{s1},\ldots,m_{sk}} \right) \binom{m_1}{m_{10},\ldots, m_{1k}}  \\
&& \times \, \left(\prod_{i=1}^k \frac{1}{(2k+\sum_{s=1}^n sm_{si}-i)!} \right)
\prod_{1\leq i<j\leq k} \left(\sum_{s=1}^n sm_{sj} - \sum_{s=1}^n sm_{si}-j+i\right).
\eeas
\end{prop}

The following is our first main theorem. It concerns joint moments of the $n_1$-th and $n_2$-th derivatives of characteristic polynomials of CUE matrices. 
Theorems \ref{intro:thm1} and \ref{intro:prop1} follow directly by setting $n_1=2, n_2=0$ and switching the roles of $k-M$ and $M$. 

\begin{thm}
\label{518mixmomentofLamda}

For $k\geq 1$, $0\leq M\leq k$ and $n_{1}\geq n_{2}\geq 0$ be integers. Let 
$P=\#\{\a_{i}=(a_{i1},\ldots,a_{in_{1}}):\sum_{j=1}^{n_{1}}ja_{ij}=n_{1},a_{ij}\geq 0\}$ and $Q=\#\{\b_{i}=(b_{i1},\ldots,b_{in_{2}}):\sum_{j=1}^{n_{2}}jb_{ij}=n_{2},b_{ij}\geq 0\}$. Denote by $\a_{i}!=\prod_{j=1}^{n_{1}}a_{ij}!$
and
$\b_{i}!=\prod_{j=1}^{n_{2}}b_{ij}!$.
Then we have
\beas
&&\int_{\U(N)} |\Lambda_A^{(n_{1})} (1)|^{2k-2M} |\Lambda_A^{(n_{2})}(1)|^{2M} dA_N\\
&=&N^{k^2+2(k-M)n_{1}+2Mn_{2}} (-1)^{(k-M)n_{1}+Mn_{2}+\frac{k(k-1)}{2}}\\ 
&& \times \,\sum_{\substack{l_1+\cdots+l_P=k-M \\ l'_1+\cdots+l'_P=k-M}}\sum_{t_{1}=0}^{\sum_{i=1}^P l'_ia_{i1}}  \binom{k-M}{l_1,\ldots,l_P} \binom{k-M}{l'_1,\ldots,l'_P} \binom{\sum_{i=1}^P l'_ia_{i1}}{t_{1}}\frac{(n_{1}!)^{2k-2M}}{\prod_{i=1}^P (\a_{i}!\prod_{j=1}^{n_{1}} j^{a_{ij}})^{l_i+l'_i}}\\
&&\times \,\sum_{\substack{u_1+\cdots+u_Q=M \\ u'_1+\cdots+u'_Q=M}}\sum_{t_{2}=0}^{\sum_{i=1}^Q u'_ib_{i1}}  \binom{M}{u_1,\ldots,u_Q} \binom{M}{u'_1,\ldots,u'_Q} \binom{\sum_{i=1}^Q u'_ib_{i1}}{t_{2}}\frac{(n_{2}!)^{2M}}{\prod_{i=1}^Q (\b_{i}!\prod_{j=1}^{n_{2}}j^{b_{ij}})^{u_i+u'_i}}\\
&&\times  \hspace{-.3cm} \sum_{ \substack{s=2,\ldots,n_{1} \\ \sum_{i=1}^k h_{si} = V_{s} } } \hspace{-.1cm} 
\prod_{s=2}^{n_{1}} \binom{V_{s}}{h_{s1}, \ldots, h_{sk}} 
\left( \frac{d}{dx} \right)^{V_{1}} \Bigg( e^{-x}x^{-\frac{k^2}{2} - (k-M)n_{1}-Mn_{2}+ \frac{1}{2} \sum_{i=1}^P (l_i+l'_i)a_{i1}+\frac{1}{2}\sum_{i=1}^{Q}(u_{i}+u'_{i})b_{i1} }\\
&& \times \, \det_{k\times k} \Big( I_{\sum_{s=2}^{n_{1}} sh_{si}+i+j-1 } (2\sqrt{x}) \Big) \Bigg) \Bigg|_{x=0}+O(N^{k^2+2(k-M)n_{1}+2Mn_{2}-1}),
\eeas
where $V_{1}=\sum_{i=1}^P (l_i+l_{i}')a_{i1}+\sum_{i=1}^Q (u_i+u'_{i})b_{i1}-t_{1}-t_{2}$, $V_{j}=\sum_{i=1}^P (l_i+l'_{i})a_{ij}+\sum_{i=1}^Q (u_i+u'_{i})b_{ij}$ for $j=2,\ldots,n_{2}$, and $V_{j}=\sum_{i=1}^P (l_i+l'_{i})a_{ij}$ for $j=\max(2,n_{2}+1),\ldots,n_{1}$.

Moreover, 
\beas
&& \hspace{-.3cm} \sum_{ \substack{s=2,\ldots,n_{1} \\ \sum_{i=1}^k h_{si} = V_{s} } }
\prod_{s=2}^{n_{1}} \binom{V_{s}}{h_{s1}, \ldots, h_{sk}} 
\left( \frac{d}{dx} \right)^{V_{1}} \Bigg(e^{-x}x^{-\frac{k^2}{2} - (k-M)n_{1}-Mn_{2}+ \frac{1}{2} \sum_{i=1}^P (l_i+l'_i)a_{i1}+\frac{1}{2}\sum_{i=1}^{Q}(u_{i}+u'_{i})b_{i1} }\\
&& \times \, \det_{k\times k} \Big( I_{\sum_{s=2}^{n_{1}} sh_{si}+i+j-1 } (2\sqrt{x}) \Big) \Bigg)  \Bigg|_{x=0} \\
&=& \sum_{ \substack{s=2,\ldots,n_{1} \\ \sum_{i=1}^k h_{si} = V_{s} \\ \sum_{i=0}^k h_{1i} = V_{1} } }(-1)^{h_{10}} \Bigg(\prod_{s=2}^{n_{1}} \binom{V_{s}}{h_{s1}, \ldots, h_{sk}} \Bigg)\binom{V_1}{h_{10},\ldots, h_{1k}} 
 \left(\prod_{i=1}^k \frac{1}{(2k+\sum_{s=1}^{n_1} sh_{si}-i)!} \right)\\
&&\times \,\prod_{1\leq i<j\leq k} \left(\sum_{s=1}^{n_1} sh_{sj} - \sum_{s=1}^{n_1} sh_{si}-j+i \right) . 
\eeas

\end{thm}

\begin{proof}
Let $\alpha=(\alpha_1,\ldots,\alpha_{2k})$.
By Lemmas \ref{15} and \ref{11/10lemma1}, 
\beas
&& \int_{\U(N)} |\Lambda_A^{(n_{1})} (1)|^{2k-2M} |\Lambda_A ^{(n_{2})}(1)|^{2M} dA_N \nonumber \\
&=& (-1)^{n_{1}(k-M)+Mn_{2}} \frac{1}{k!(2\pi \i)^k} \prod_{j=1}^{k-M} (\frac{d}{d\alpha_j} \frac{d}{d\alpha_{k+j}} )^{n_{1}}\prod_{j=k-M+1}^{k} (\frac{d}{d\alpha_j} \frac{d}{d\alpha_{k+j}} )^{n_{2}} \nonumber  \\
&&\times \,\int_{|w_i|=1} e^{N \sum_{j=1}^k (w_j-\alpha_{j})} \prod_{i\neq j} z(w_i-w_j)^{-1} \prod_{\substack{ 1\leq i \leq k\\ 1 \leq j \leq 2k }} z(w_i-\alpha_j)\prod_{j=1}^{k} dw_j \Bigg|_{\alpha=0} \nonumber \\ 
&& +\, O(N^{k^2+(2k-2M)n_{1}+2Mn_{2}-1}).
\eeas
Suppose that $\alpha_j=\alpha_j(N)$ and $|\alpha_j| \leq 1/N$ for $j=1,2,\ldots,2k$ above, we have furthermore that the last expression is
\beas
&=& (-1)^{n_{1}(k-M)+n_{2}M+\frac{k(k-1)}{2}} \frac{1}{k!(2\pi \i)^k}\prod_{j=1}^{k-M}  \left( \frac{d}{d\alpha_j} \frac{d}{d\alpha_{k+j}} \right)^{n_{1}} \prod_{j=k-M+1}^{k} \left(\frac{d}{d\alpha_j} \frac{d}{d\alpha_{k+j}} \right)^{n_{2}} \nonumber \\
&&\times \,\int_{|w_i|=1} \frac{e^{N \sum_{j=1}^k(w_j-\alpha_{j})} \Delta^2(w)}{\prod_{j=1}^{2k} \prod_{i=1}^k(w_i-\alpha_j)} \prod_{j=1}^{k} dw_j \Bigg|_{\alpha=0}
+\, O(N^{k^2+(2k-2M)n_{1}+2Mn_{2}-1}) .
\eeas
By Lemma \ref{lem1},
\bea
&&\prod_{j=1}^{k-M}  \left( \frac{d}{d\alpha_j} \frac{d}{d\alpha_{k+j}} \right)^{n_{1}} \prod_{j=k-M+1}^{k} (\frac{d}{d\alpha_j} \frac{d}{d\alpha_{k+j}} )^{n_{2}}\int_{|w_i|=1} \frac{e^{N \sum_{j=1}^kw_j} \Delta^2(w)}{\prod_{j=1}^{2k} \prod_{i=1}^k(w_i-\alpha_j)} e^{-N \sum_{j=1}^{k} \alpha_j} \prod_{j=1}^{k} dw_j \Bigg|_{\alpha=0}\nonumber\\
=&&\int_{|w_i|=1} \frac{e^{N \sum_{j=1}^kw_j} \Delta^2(w)}{ \prod_{i=1}^kw_i^{2k}} \Bigg( \sum_{m_1+2m_2+\cdots+n_{1}m_{n_{1}}=n_{1}} \frac{n_{1}!}{m_1!\cdots m_{n_{1}}!} \left( - N + \sum_{l=1}^k \frac{1}{w_l} \right)^{m_1} 
\nonumber \\
&&\times \, \prod_{j=2}^{n_{1}} \left(\frac{1}{j} \sum_{l=1}^k \frac{1}{w_l^j} \right)^{m_j}  \Bigg)^{k-M} \Bigg( \sum_{r_1+2r_2+\cdots+n_{2}r_{n_{2}}=n_{2}} \frac{n_{2}!}{r_1!\cdots r_{n_{2}}!} \left( - N + \sum_{l=1}^k \frac{1}{w_l} \right)^{r_{1}}\nonumber \\
&&\times \,\prod_{j=2}^{n_{2}} \left(\frac{1}{j} \sum_{l=1}^k \frac{1}{w_l^j} \right)^{r_j}  \Bigg)^{M}\Bigg( \sum_{m_1+2m_2+\cdots+n_{1}m_{n_{1}}=n_{1}} \frac{n_{1}!}{m_1!\cdots m_{n_{1}}!} \prod_{j=1}^{n_{1}} \left(\frac{1}{j} \sum_{l=1}^k \frac{1}{w_l^j} \right)^{m_j}  \Bigg)^{k-M}\nonumber \\
&&\times \,\Bigg( \sum_{r_1+2r_2+\cdots+n_{2}r_{n_{2}}=n_{2}} \frac{n_{2}!}{r_1!\cdots r_{n_{2}}!} \prod_{j=1}^{n_{2}} \left(\frac{1}{j} \sum_{l=1}^k \frac{1}{w_l^j} \right)^{r_j}  \Bigg)^{M}
 \prod_{j=1}^k dw_j.
\label{518mixmomentofLamda:eq}
\eea
Expanding all the powers involved, we obtain 
\bea
&=& \int_{|w_i|=1} \frac{e^{N \sum_{j=1}^kw_j} \Delta^2(w)}{ \prod_{i=1}^kw_i^{2k}} \Bigg(\sum_{l_1+\cdots+l_P=k-M} \binom{k-M}{l_1,\ldots,l_P} \frac{(n_{1}!)^{k-M}}{\prod_{i=1}^P (\a_i!)^{l_i}} \left( - N + \sum_{l=1}^k \frac{1}{w_l} \right)^{\sum_{i=1}^P l_ia_{i1}} \nonumber\\
&&\times \,\prod_{j=2}^{n_{1}} \left(\frac{1}{j} \sum_{l=1}^k \frac{1}{w_l^j} \right)^{\sum_{i=1}^P l_ia_{ij}}\Bigg)
\Bigg(\sum_{u_1+\cdots+u_Q=M} \binom{M}{u_1,\ldots,u_Q} \frac{(n_{2}!)^{M}}{\prod_{i=1}^Q (\b_i!)^{u_i}} \left( -N + \sum_{l=1}^k \frac{1}{w_l} \right)^{\sum_{i=1}^Q u_ib_{i1}} \nonumber\\
&&\times \,\prod_{j=2}^{n_{2}} \left(\frac{1}{j} \sum_{l=1}^k \frac{1}{w_l^j} \right)^{\sum_{i=1}^Q u_ib_{ij}}\Bigg) \Bigg(\sum_{l'_1+\cdots+l'_P=k-M} \binom{k-M}{l'_1,\ldots,l'_P} \frac{(n_{1}!)^{k-M}}{\prod_{i=1}^P (\a_i!)^{l'_i}} \prod_{j=1}^{n_{1}} \left(\frac{1}{j} \sum_{l=1}^k \frac{1}{w_l^j} \right)^{\sum_{i=1}^P l'_ia_{ij}}\Bigg) \nonumber \\
&&\times \,
\Bigg(\sum_{u'_1+\cdots+u'_Q=M} \binom{M}{u'_1,\ldots,u'_Q} \frac{(n_{2}!)^{M}}{\prod_{i=1}^Q (\b_i!)^{u'_i}}
\prod_{j=1}^{n_{2}} \left(\frac{1}{j} \sum_{l=1}^k \frac{1}{w_l^j} \right)^{\sum_{i=1}^Q u'_ib_{ij}}\Bigg)
\prod_{j=1}^{k}dw_{j}.
\label{516night1}
\eea
Next, we expand $\left(\sum_{l=1}^k \frac{1}{w_l} \right)^{\sum_{i=1}^P l'_ia_{i1}}$ and $\left(\sum_{l=1}^k \frac{1}{w_l} \right)^{\sum_{i=1}^Q u'_ib_{i1}}$
by viewing $\sum_{l=1}^k \frac{1}{w_l}-N$ as a whole. We then simplify the expression by viewing $\sum_{l=1}^k \frac{1}{w_l^j}$ as a whole for $j=2,\ldots,n_1$, and obtain
\beas
(\ref{516night1}) &=& \sum_{\substack{l_1+\cdots+l_P=k-M \\ l'_1+\cdots+l'_P=k-M}} \binom{k-M}{l_1,\ldots,l_P} \binom{k-M}{l'_1,\ldots,l'_P} \frac{(n_{1}!)^{2k-2M}}{\prod_{i=1}^P (\prod_{j=1}^{n_{1}} a_{ij}! j^{a_{ij}})^{l_i+l'_i}}\sum_{t_{1}=0}^{\sum_{i=1}^P l'_ia_{i1}} \binom{\sum_{i=1}^P l'_ia_{i1}}{t_{1}}N^{t_{1}}\\
&& \hspace{-.7cm} \times \, \sum_{\substack{u_1+\cdots+u_Q=M \\ u'_1+\cdots+u'_Q=M}} \binom{M}{u_1,\ldots,u_Q} \binom{M}{u'_1,\ldots,u'_Q} \frac{(n_{2}!)^{2M}}{\prod_{i=1}^Q (\prod_{j=1}^{n_{2}} b_{ij}! j^{b_{ij}})^{u_i+u'_i}}\sum_{t_{2}=0}^{\sum_{i=1}^Q u'_ib_{i1}} \binom{\sum_{i=1}^Q u'_ib_{i1}}{t_{2}}N^{t_{2}}\\
&&\hspace{-.7cm} \times \, \int_{|w_i|=1} \frac{e^{N \sum_{j=1}^k w_j} \Delta^2(w)}{\prod_{i=1}^k w_i^{2k}}(\sum_{l=1}^{k}\frac{1}{w_{l}}-N)^{V_{1}}\prod_{j=2}^{n_{1}}\left( \sum_{l=1}^k \frac{1}{w_l^j}\right)^{V_{j}}\prod_{j=1}^k dw_j,
\eeas
where the quantities $V_j$ are defined in the theorem.
Finally, we apply Proposition \ref{prop18} to the last integral with $m_j=V_j$ for $j=1,\ldots,n_1$.
In this process, we can compute that  $\sum_{s=2}^{n_1} s V_s=\sum_{s=1}^{n_1} s V_s - V_1 = 2 (k-M)n_{1} +2Mn_{2} - \sum_{i=1}^P (l_i+l'_i)a_{i1} - \sum_{i=1}^{Q}(u_{i}+u'_{i})b_{i1}$. Hence, we obtain the claimed result. 
\end{proof}

We present below our main theorem on joint moments of the $n_1$-th and $n_2$-th derivatives of the analogue of Hardy's $Z$-function.  Theorems \ref{into:thm2} and \ref{into:prop2} are obtained by choosing $n_1=2,n_2=0$ and switching $n-M$ and $M$.

\begin{thm}
\label{generalization of mixed moment}
For $k\geq 1$, $0\leq M\leq k$ and $n_{1}\geq n_{2}\geq 0$ be integers.  Let $P=\#\{\a_{i}=(a_{i1},\ldots,a_{in_{1}}):\sum_{j=1}^{n_{1}}ja_{ij}=n_{1},a_{ij}\geq 0\}$ and $Q=\#\{\b_{i}=(b_{i1},\ldots,b_{in_{2}}):\sum_{j=1}^{n_{2}}jb_{ij}=n_{2},b_{ij}\geq 0\}$. Denote by $\a_{i}!=\prod_{j=1}^{n_{1}}a_{ij}!$
and
$\b_{i}!=\prod_{j=1}^{n_{2}}b_{ij}!$.
Then we have
\beas
&&\int_{\U(N)} |Z_A^{(n_{1})} (1)|^{2k-2M} |Z_A^ {(n_{2})}(1)|^{2M} dA_N\\
&=&(-1)^{n_{1}(k-M)+n_{2}M+\frac{k(k-1)}{2}}N^{k^2+2n_{1}(k-M)+2n_{2}M}\\ 
&&\times \, \sum_{\substack{l_1+\cdots+l_P= 2k-2M\\u_{1}+\cdots+u_{Q}=2M}}\binom{2k-2M}{l_1,\ldots,l_P}\binom{2M}{u_1,\ldots,u_Q}
\frac{(n_{1}!)^{2k-2M}}{(\prod_{i=1}^P(\a_{i}!)^{l_{i}})(\prod_{j=1}^{n_{1}} j^{\sum_{i=1}^{P}l_{i}a_{ij}})} \\
&&\times \, \frac{(n_{2}!)^{2M}}{(\prod_{i=1}^Q (\b_{i}!)^{u_{i}})(\prod_{j=1}^{n_{2}} j^{\sum_{i=1}^{Q}u_{i}b_{ij}} )}
\sum_{\substack{s=2,\ldots,n_{1} \\ \sum_{i=1}^k h_{si}=V_{s}}} 
\prod_{s=2}^{n_{1}} \binom{V_{s}}{h_{s1},\ldots,h_{sk}}
\\
&& \times \, \left( \frac{d}{dx} \right)^{V_{1}} \left( e^{-\frac{x}{2}} x^{-\frac{k^2}{2} - (k-M)n_{1}-Mn_{2}+ \frac{1}{2} V_{1}}\det_{k\times k} \Big( I_{\sum_{s=2}^{n_{1}} s h_{si}+i+j-1 } (2\sqrt{x}) \Big) \right) \Bigg|_{x=0}\\ 
&&+ \, O(N^{k^2+2n_{1}(k-M)+2n_{2}M-1}),
\eeas
where
$V_{j}=\sum_{i=1}^{P}l_{i}a_{i,j}+\sum_{i=1}^{Q}u_{i}b_{i,j}$ for $j=1,\ldots,n_{2}$ and $V_{j}=\sum_{i=1}^{P}l_{i}a_{i,j}$ for $j=n_{2}+1,\ldots,n_{1}$.

Moreover,
\beas
&& 
\sum_{\substack{s=2,\ldots,n_{1} \\ \sum_{i=1}^k h_{si}=V_{s}}} \left(\prod_{s=2}^{n_{1}} \binom{V_{s}} {h_{s1},\ldots,h_{sk}}\right) 
\left( \frac{d}{dx} \right)^{V_{1}} \Bigg(e^{-\frac{x}{2}} x^{-\frac{k^2}{2} - (k-M)n_{1}-Mn_{2}+ \frac{1}{2} V_{1}} \\
&& \times \, \det_{k\times k} \Big( I_{\sum_{s=2}^{n_{1}} s h_{si}+i+j-1 } (2\sqrt{x}) \Big) \Bigg) \Bigg|_{x=0} \\
&=& \sum_{\substack{
\sum_{i=1}^k h_{si} = V_s \\ s=2,\ldots,n_1 \\
\sum_{i=0}^k h_{1i}=V_1}} 
 (-\frac{1}{2})^{h_{10}}
\left(\prod_{s=2}^{n_1} \binom{V_s}{h_{s1},\ldots,h_{sk}} \right) \binom{V_1}{h_{10},\ldots, h_{1k}}  \\
&& \times \, \left(\prod_{i=1}^k \frac{1}{(2k+\sum_{s=1}^{n_{1}} sh_{si}-i)!} \right)
\prod_{1\leq i<j\leq k} \left(\sum_{s=1}^{n_1} sh_{sj} - \sum_{s=1}^{n_1} sh_{si}-j+i\right).
\eeas
\end{thm}

\begin{proof}
Let $\alpha=(\alpha_1,\ldots,\alpha_{2k})$.
By Lemmas \ref{15} and \ref{11/10lemma1} , we have
\beas
&& \int_{\U(N)} |Z_A^{(n_{1})} (1)|^{2k-2M} |Z_A ^{(n_{2})}(1)|^{2M} dA_N \nonumber \\
&=& (-1)^{n_{1}(k-M)+Mn_{2}} \frac{1}{k!(2\pi \i)^k} \prod_{j=1}^{k-M} \left(\frac{d}{d\alpha_j} \frac{d}{d\alpha_{k+j}} \right)^{n_{1}}\prod_{j=k-M+1}^{k} \left(\frac{d}{d\alpha_j} \frac{d}{d\alpha_{k+j}} \right)^{n_{2}} \nonumber \\
&&\times \, \int_{|w_i|=1} e^{N \sum_{j=1}^k w_j} \prod_{i\neq j} z(w_i-w_j)^{-1} \prod_{\substack{ 1\leq i \leq k\\ 1 \leq j \leq 2k }} z(w_i-\alpha_j)  \nonumber \\
&& \times \, e^{-\frac{N}{2} \sum_{j=1}^{2k} \alpha_j} \prod_{j=1}^{k} dw_j \Bigg|_{\alpha=0} + O(N^{k^2+(2k-2M)n_{1}+2Mn_{2}-1}).
\eeas
Suppose that $\alpha_j=\alpha_j(N)$ and $|\alpha_j| \leq 1/N$ for $j=1,2,\ldots,2k$ above, the right-hand side is 
\beas
&=& (-1)^{n_{1}(k-M)+n_{2}M+\frac{k(k-1)}{2}} \frac{1}{k!(2\pi \i)^k}\prod_{j=1}^{k-M}  \left( \frac{d}{d\alpha_j} \frac{d}{d\alpha_{k+j}} \right)^{n_{1}} \prod_{j=k-M+1}^{k} \left(\frac{d}{d\alpha_j} \frac{d}{d\alpha_{k+j}} \right)^{n_{2}} \nonumber \\
&&\times \, \int_{|w_i|=1} \frac{e^{N \sum_{j=1}^kw_j} \Delta^2(w)}{\prod_{j=1}^{2k} \prod_{i=1}^k(w_i-\alpha_j)} e^{-\frac{N}{2} \sum_{j=1}^{2k} \alpha_j} \prod_{j=1}^{k} dw_j \Bigg|_{\alpha=0}
+\, O(N^{k^2+(2k-2M)n_{1}+2Mn_{2}-1}) .
\eeas
Using Lemma \ref{lem1} to the derivatives with respect to $\alpha_j$ in the above formula, we have
\bea
&=& \frac{1}{k!(2\pi \i)^{k}}(-1)^{n_{1}(k-M)+n_{2}M+\frac{k(k-1)}{2}} \int_{|w_i|=1} \frac{e^{N \sum_{j=1}^kw_j} \Delta^2(w)}{ \prod_{i=1}^kw_i^{2k}} \Bigg( \sum_{m_1+2m_2+\cdots+n_{1}m_{n_{1}}=n_{1}} \frac{n_{1}!}{m_1!\cdots m_{n_{1}}!}\nonumber\\
&&\times \, \left( - \frac{N}{2} + \sum_{l=1}^k \frac{1}{w_l} \right)^{m_1} 
\prod_{j=2}^{n_{1}} \left(\frac{1}{j} \sum_{l=1}^k \frac{1}{w_l^j} \right)^{m_j}  \Bigg)^{2k-2M}\Bigg( \sum_{r_1+2r_2+\cdots+n_{2}r_{n_{2}}=n_{2}} \frac{n_{2}!}{r_1!\cdots r_{n_{2}}!}\nonumber\\
&&\times \, \left( - \frac{N}{2} + \sum_{l=1}^k \frac{1}{w_l} \right)^{r_1} 
\prod_{j=2}^{n_{2}} \left(\frac{1}{j} \sum_{l=1}^k \frac{1}{w_l^j} \right)^{r_j}  \Bigg)^{2M} 
 \prod_{j=1}^k dw_j + O(N^{k^2+(2k-2M)n_{1}+2Mn_{2}-1}) .
\label{generalization of mixed moment:eq}
\eea
Expanding the powers, we write
\beas
&& \Bigg( \sum_{m_1+2m_2+\cdots+n_{1}m_{n_{1}}=n_{1}} \frac{n_{1}!}{m_1!\cdots m_{n_{1}}!} \left( - \frac{N}{2} + \sum_{l=1}^k \frac{1}{w_l} \right)^{m_1} 
\prod_{j=2}^{n_{1}}\left(\frac{1}{j} \sum_{l=1}^k \frac{1}{w_l^j} \right)^{m_j}  \Bigg)^{2k-2M} \\
&=& \sum_{l_1+\cdots+l_P=2k-2M} \binom{2k-2M}{l_1,\ldots,l_P} \frac{(n_{1}!)^{2k-2M}}{\prod_{i=1}^P (\a_i!)^{l_i}} \left( - \frac{N}{2} + \sum_{l=1}^k \frac{1}{w_l} \right)^{\sum_{i=1}^P l_ia_{i1}} 
\prod_{j=2}^{n_{1}} \left(\frac{1}{j} \sum_{l=1}^k \frac{1}{w_l^j} \right)^{\sum_{i=1}^P l_ia_{ij}},
\eeas
and
\beas
&&\Bigg( \sum_{r_1+2r_2+\cdots+n_{2}r_{n_{2}}=n_{2}} \frac{n_{2}!}{r_1!\cdots r_{n_{2}}!}\left( - \frac{N}{2} + \sum_{l=1}^k \frac{1}{w_l} \right)^{r_1} 
\prod_{j=2}^{n_{2}} \left(\frac{1}{j} \sum_{l=1}^k \frac{1}{w_l^j} \right)^{r_j}  \Bigg)^{2M}\\
&=&\sum_{u_1+\cdots+u_Q=2M} \binom{2M}{u_1,\ldots,u_Q} \frac{(n_{2}!)^{2M}}{\prod_{i=1}^Q (\b_i!)^{u_i}} \left( - \frac{N}{2} + \sum_{l=1}^k \frac{1}{w_l} \right)^{\sum_{i=1}^Q u_ib_{i1}} 
\prod_{j=2}^{n_{2}} \left(\frac{1}{j} \sum_{l=1}^k \frac{1}{w_l^j} \right)^{\sum_{i=1}^Q u_ib_{ij}}.
\eeas
The claim in this theorem now follows from Proposition \ref{prop1} with $m_{j}=\sum_{i=1}^{P}l_{i}a_{i,j}+\sum_{i=1}^{Q}u_{i}b_{i,j}$ for $j=1,\ldots,n_{2}$ and $m_{j}=\sum_{i=1}^{P}l_{i}a_{i,j}$ for $j=n_{2}+1,\ldots,n_{1}$.
Here we can compute that
$\sum_{s=2}^{n_1} s m_s=\sum_{s=1}^{n_1} s m_s - m_1 = 2(k-M)n_{1}+ 2Mn_{2} - m_{1}$.
\end{proof}

The following is our third main result. It provides an alternative formula for Theorem \ref{518mixmomentofLamda}. Compared with Theorem \ref{518mixmomentofLamda}, this theorem provides an effective approach to compute the joint moments when $k$ is small and $n_1, n_2$ are large. We will explain this in detail after the proof of the theorem. Theorem \ref{intro:thm3} is obtained by setting $n_1=n,n_2=0$ and switching $k-M$ and $M$.

\begin{thm}
\label{general mixed moment on Lambda}
For $k\geq 1$, $0\leq M\leq k$ and $n_{1},n_{2}\geq 0$ be integers, then we have
\beas
&&\int_{\U(N)} |\Lambda_A^{(n_{1})} (1)|^{2k-2M} |\Lambda_A ^{(n_{2})}(1)|^{2M} dA_N \\
&=& (-1)^{\frac{k(k-1)}{2}}
N^{k^2+(2k-2M)n_{1}+2Mn_{2}} 
\sum_{ \substack{ \sum_{l=1}^k s_{l,q_{1}} \leq n_{1} \\ \sum_{l=1}^k h_{l,q_{1}} = n_{1} \\
q_{1}=1,\ldots,k-M} } \sum_{ \substack{\sum_{l=1}^k s'_{l,q_{2}} \leq n_{2} \\ \sum_{l=1}^k h'_{l,q_{2}} = n_{2} \\
q_{2}=1,\ldots,M} }  (-1)^{\sum_{q_{1}=1}^{k-M}\sum_{l=1}^k s_{l,q_{1}}+\sum_{q_{2}=1}^{M}\sum_{l=1}^k s'_{l,q_{2}}}
\\
&& \times \, (n_{1}!)^{2k-2M}(n_{2}!)^{2M} 
 \left(\prod_{q_{1}=1}^{k-M}\frac{1}{(n_{1}-\sum_{l=1}^{k}s_{l,q_{1}})!}\right) \left(\prod_{q_{2}=1}^{M}\frac{1}{(n_{2}-\sum_{l=1}^{k}s'_{l,q_{2}})!}\right)
\\
&& \times \, \left(\prod_{i=1}^k \frac{1}{(2k-i+W_{i})!} \right) \prod_{1\leq i<j\leq k} (W_{j}-W_{i}-j+i) +O(N^{k^2+(2k-2M)n_{1}+2Mn_{2}-1}),
\eeas
where for $i=1,\ldots,k$,
\bea\label{definition of W}
W_{i}=\sum_{q_{1}=1}^{k-M} (s_{i,q_{1}}+h_{i,q_{1}})+\sum_{q_{2}=1}^{M} (s'_{i,q_{2}}+h'_{i,q_{2}}).
\eea
\end{thm}

\begin{proof}
Following the proof of Theorem \ref{518mixmomentofLamda}, we have (\ref{518mixmomentofLamda:eq}). Instead of expanding the powers, here we apply (\ref{eq1128formula1}) and (\ref{eq1128formula2}) in Proposition \ref{prelimiary:prop1} to tackle the powers. We then obtain

\bea
(\ref{518mixmomentofLamda:eq}) 
&=& \int_{|w_i|=1} \frac{e^{N \sum_{j=1}^kw_j} \Delta^2(w)}{ \prod_{i=1}^kw_i^{2k}} \Bigg(\sum_{ \substack{\sum_{l=1}^k s_{l,q} \leq n_{1} \\ 
q=1,\ldots,k-M} } \prod_{q=1}^{k-M} \Big((-N)^{n_{1}-\sum_{l=1}^k s_{l,q}} \binom{n_{1}}{\sum_{l=1}^k s_{l,q}} (\sum_{l=1}^k s_{l,q})! \Big)\nonumber\\
&&\times \, \frac{1}{\prod_{l=1}^k w_l^{\sum_{q=1}^{k-M} s_{l,q}}} \Bigg)
\Bigg(\sum_{ \substack{\sum_{l=1}^k s'_{l,q} \leq n_{2} \\ q=1,\ldots,M} } \prod_{q=1}^{M} \Big((-N)^{n_{2}-\sum_{l=1}^k s'_{l,q}} \binom{n_{2}}{\sum_{l=1}^k s'_{l,q}} (\sum_{l=1}^k s'_{l,q})! \Big)\nonumber \\
&&\times \, \frac{1}{\prod_{l=1}^k w_l^{\sum_{q=1}^{M} s'_{l,q}}} \Bigg) \Bigg(\sum_{ \substack{ \sum_{l=1}^{k} h_{l,q} = n_1 \\ 
q=1,\ldots,k-M} }  \frac{(n_1!)^{k-M}}{\prod_{l=1}^k w_l^{\sum_{q=1}^{k-M} h_{l,q}}}\Bigg) \nonumber \\
&& \times \, \Bigg(\sum_{ \substack{\sum_{l=1}^k h'_{l,q}= n_2 \\ 
q=1,\ldots,M} } \frac{(n_2!)^{M}}{\prod_{l=1}^k w_l^{\sum_{q=1}^M h'_{l,q}}} \Bigg)
\prod_{j=1}^{k}dw_{j}.\label{516night}
\eea
Combining Proposition \ref{lemmaonthe integral},
\beas
(\ref{516night}) &=&(2\pi \i)^{k}
\sum_{ \substack{ \sum_{l=1}^k s_{l,q_{1}} \leq n_{1} \\ \sum_{l=1}^k h_{l,q_{1}} = n_{1} \\
q_{1}=1,\ldots,k-M} }
\sum_{ \substack{\sum_{l=1}^k s'_{l,q_{2}} \leq n_{2} \\ \sum_{l=1}^k h'_{l,q_{2}} = n_{2} \\
q_{2}=1,\ldots,M} }
(n_{1}!)^{k-M}(n_{2}!)^{M}\\
&&\times \, \prod_{q_{1}=1}^{k-M} \left( 
(-N)^{n_{1}-\sum_{l=1}^k s_{l,q_{1}}} \binom{n_{1}}{\sum_{l=1}^k s_{l,q_{1}}} (\sum_{l=1}^k s_{l,q_{1}})! 
\right) \\
&&\times \, \prod_{q_{2}=1}^{M} \left( 
(-N)^{n_{2}-\sum_{l=1}^k s'_{l,q_{2}}} \binom{n_{2}}{\sum_{l=1}^k s'_{l,q_{2}}} (\sum_{l=1}^k s'_{l,q_{2}})! 
\right)\\
&& \times \, \sum_{\sigma\in S_k} \det_{k\times k} \Bigg( \frac{N^{2k+W_{\sigma(i)}-1-(i+j-2)}}{(2k+W_{\sigma(i)}-1-(i+j-2))!} \Bigg).
\eeas
In the determinant, we firstly extract the factor  $N^{W_{\sigma(i)}+k+1-i}$ from the $i$-th row and secondly extract the factor $N^{k-j}$ from the $j$-th column, then
\beas
(\ref{516night}) &=& (2\pi \i)^{k}N^{k^2+2(k-M)n_{1}+2Mn_{2}}
\sum_{\sigma\in S_k} \sum_{ \substack{ \sum_{l=1}^k s_{l,q_{1}} \leq n_{1} \\ \sum_{l=1}^k h_{l,q_{1}} = n_{1} \\
q_{1}=1,\ldots,k-M} }
\sum_{ \substack{ \sum_{l=1}^k s'_{l,q_{2}} \leq n_{2} \\ \sum_{l=1}^k h'_{l,q_{2}} = n_{2} \\
q_{2}=1,\ldots,M} }
\hspace{-.4cm}
(n_{1}!)^{2k-2M}(n_{2}!)^{2M}(-1)^{(k-M)n_{1}+Mn_{2}}\\
&&\times \, (-1)^{\sum_{q_{1}=1}^{k-M}\sum_{l=1}^k s_{l,q_{1}}+\sum_{q_{2}=1}^{M}\sum_{l=1}^k s'_{l,q_{2}}}
\left(\prod_{q_{1}=1}^{k-M}\frac{1}{(n_{1}-\sum_{l=1}^{k}s_{l,q_{1}})!}\right)
\left(\prod_{q_{2}=1}^{M}\frac{1}{(n_{2}-\sum_{l=1}^{k}s'_{l,q_{2}})!} \right) \\
&& \times \,\det_{k\times k} \Bigg( \frac{1}{(2k+1+W_{\sigma(i)}-i-j)!} \Bigg).
\eeas
It is not hard to check that the summations over $s_{i,q_{1}}, h_{i,q_{1}}, s'_{i,q_{2}}, h'_{i,q_{2}}$ are invariant under any permutation $\sigma$ of $i$ on the set $\{1,\ldots,k \}$. Indeed, for a fixed $\sigma$, we change the variables $s_{\sigma(i),q_{1}}=\tilde{s}_{i,q_{1}}, h_{\sigma(i),q_{1}}=\tilde{h}_{i,q_{1}}$,
$s'_{\sigma(i),q_{2}}=\tilde{s'}_{i,q_{2}},
h'_{\sigma(i),q_{2}}=\tilde{h'}_{i,q_{2}}$, then $\sum_{l=1}^k s_{l,q_{1}} = \sum_{l=1}^k \tilde{s}_{l,q_{1}}$, $\sum_{l=1}^k s'_{l,q_{2}} = \sum_{l=1}^k \tilde{s}'_{l,q_{2}}$, $\sum_{l=1}^k h_{l,q_{1}} = \sum_{l=1}^k \tilde{h}_{l,q_{1}}$, $\sum_{l=1}^k h'_{l,q_{2}} = \sum_{l=1}^k \tilde{h}'_{l,q_{2}}$. We also have $W_{\sigma(i)}=\sum_{q_{1}=1}^{k-M} (s_{\sigma(i),q_{1}}+h_{\sigma(i),q_{1}})+\sum_{q_{2}=1}^{M} (s'_{\sigma(i),q_{2}}+h'_{\sigma(i),q_{2}})
=\sum_{q_{1}=1}^{k-M} (\tilde{s}_{i,q_{1}}+\tilde{h}_{i,q_{1}})+\sum_{q_{2}=1}^{M} (\tilde{s}'_{i,q_{2}}+\tilde{h}'_{i,q_{2}}).$ So the whole expression is independent of $\sigma$ and we can set $\sigma=\text{id}$. Hence, $\sum_{\sigma\in S_k}$ can be replaced $k!$.

By Lemma \ref{determinant},
\beas
(\ref{516night}) &=& k!(2\pi \i)^{k}N^{k^2+2(k-M)n_{1}+2Mn_{2}}
\sum_{ \substack{\sum_{l=1}^k s_{l,q_{1}} \leq n_{1} \\ \sum_{l=1}^k h_{l,q_{1}} = n_{1} \\
q_{1}=1,\ldots,k-M} } 
\sum_{ \substack{ \sum_{l=1}^k s'_{l,q_{2}} \leq n_{2} \\ \sum_{l=1}^k h'_{l,q_{2}} = n_{2} \\
q_{2}=1,\ldots,M} } \hspace{-.2cm}
(n_{1}!)^{2k-2M}(n_{2}!)^{2M}(-1)^{(k-M)n_{1}+Mn_{2}}\\
&& \times \, (-1)^{\sum_{q_{1}=1}^{k-M}\sum_{l=1}^k s_{l,q_{1}}+\sum_{q_{2}=1}^{M}\sum_{l=1}^k s'_{l,q_{2}}}\prod_{q_{1}=1}^{k-M}\frac{1}{(n_{1}-\sum_{l=1}^{k}s_{l,q_{1}})!}\prod_{q_{2}=1}^{M}\frac{1}{(n_{2}-\sum_{l=1}^{k}s'_{l,q_{2}})!}\\
&& \times \, \prod_{i=1}^k \frac{1}{(2k-i+W_{i})!}\prod_{1\leq i<j\leq k} (W_{j}-W_{i}-j+i),
\eeas
where $W_{i}$ is given in (\ref{definition of W}).
Hence, we obtain the claimed result in this theorem.
\end{proof}

\begin{rem}
\label{remark: the difference}
{\rm We now explain the difference between Theorem \ref{518mixmomentofLamda} and Theorem \ref{general mixed moment on Lambda}.
When using Theorem \ref{518mixmomentofLamda}, it is required to list all tuples of the sets for the two combinatorial objects $P, Q$. This may not be an easy task when $n_1, n_2$ are large. In comparison, the summation in Theorem \ref{general mixed moment on Lambda} is easier when $k$ is small. So Theorem \ref{general mixed moment on Lambda} is computationally effective. However, the formula in Theorem \ref{518mixmomentofLamda} has a nice structure depending on determinants of matrices whose entries involve the modified Bessel functions of the first kind, which can be used to build connections between moments and Painlev\'{e} equations. This is studied in detail in our second paper \cite{keating-fei}.
}\end{rem}

The following is our fourth main result. It provides an alternative formula for Theorem \ref{generalization of mixed moment}. Theorem \ref{intro:thm4} is obtained from it by choosing $n_1=n, n_2=0$ and switching the positions of $k-M$ and $M$.

\begin{thm}\label{general 1128mixed moment}
For $k\geq 1$, $0\leq M\leq k$ and $n_{1},n_{2}\geq 0$ be integers, we have
\beas
&&\int_{\U(N)} |Z_A^{(n_{1})} (1)|^{2k-2M} |Z_A ^{(n_{2})}(1)|^{2M} dA_N \\
&=& (-1)^{n_{1}(k-M)+n_{2}M+\frac{k(k-1)}{2}}
N^{k^2+(2k-2M)n_{1}+2Mn_{2}}
\sum_{ \substack{\sum_{l=1}^k s_{l,q_{1}} \leq n_{1} \\ \sum_{l=1}^k h_{l,q_{2}} \leq n_{2}\\
q_{1}=1,\ldots,2k-2M\\
q_{2}=1,\ldots,2M}}(n_{1}!)^{2k-2M}(n_{2}!)^{2M}  \\
&& \times \, (-1/2)^{2kn_{1}-2Mn_{1}-\sum_{q_{1}=1}^{2k-2M}\sum_{l=1}^k s_{l,q_{1}}} (-1/2)^{2Mn_{2}-\sum_{q_{2}=1}^{2M}\sum_{l=1}^k h_{l,q_{2}}}
\left(\prod_{q_{1}=1}^{2k-2M}\frac{1}{(n_{1}-\sum_{l=1}^{k}s_{l,q_{1}})!}\right) \\
&& \times \, \left(\prod_{q_{2}=1}^{2M}\frac{1}{(n_{2}-\sum_{l=1}^{k}h_{l,q_{2}})!} \right)
\left(\prod_{i=1}^k \frac{1}{(2k-i+\sum_{q_{1}=1}^{2k-2M} s_{i,q_{1}}+\sum_{q_{2}=1}^{2M} h_{i,q_{2}})!} \right) \\
&& \times \, \prod_{1\leq i<j\leq k} \left( \Big(\sum_{q_{1}=1}^{2k-2M} s_{j,q_{1}}+\sum_{q_{2}=1}^{2M}h_{j,q_{2}}\Big) - \Big(\sum_{q_{1}=1}^{2k-2M}s_{i,q_{1}}+\sum_{q_{2}=1}^{2M}h_{i,q_{2}}\Big) -j+i\right)\\
&&+ \, O(N^{k^2+(2k-2M)n_{1}+2Mn_{2}-1}).
\eeas
\end{thm}

\begin{proof}
Following the proof of Theorem \ref{generalization of mixed moment}, we have (\ref{generalization of mixed moment:eq}). We now apply (\ref{eq1128formula1}) in Proposition \ref{prelimiary:prop1} to tackle the powers and obtain
\beas
(\ref{generalization of mixed moment:eq}) &=& \frac{1}{k!(2\pi \i)^{k}}(-1)^{n_{1}(k-M)+n_{2}M+\frac{k(k-1)}{2}} \int_{|w_i|=1} \frac{e^{N \sum_{j=1}^kw_j} \Delta^2(w)}{ \prod_{i=1}^kw_i^{2k}} \nonumber \\
&& 
\times \, \Bigg(\sum_{ \substack{ \sum_{l=1}^k s_{l,q} \leq n_{1} \\ 
q=1,\ldots,2k-2M} } \prod_{q=1}^{2k-2M} \left((-N/2)^{n_{1}-\sum_{l=1}^k s_{l,q}} \binom{n_{1}}{\sum_{l=1}^k s_{l,q}} (\sum_{l=1}^k s_{l,q})! \right)\frac{1}{\prod_{l=1}^k w_l^{\sum_{q=1}^{2k-2M} s_{l,q}}} \Bigg)
\nonumber \\
&&
\times \, \Bigg(\sum_{ \substack{\sum_{l=1}^k h_{l,q} \leq n_{2} \\ q=1,\ldots,2M} }\prod_{q=1}^{2M} \left((-N/2)^{n_{2}-\sum_{l=1}^k h_{l,q}} \binom{n_{2}}{\sum_{l=1}^k h_{l,q}} (\sum_{l=1}^k h_{l,q})! \right)\frac{1}{\prod_{l=1}^k w_l^{\sum_{q=1}^{2M} h_{l,q}}} \Bigg) \nonumber \\
&&
\times \, \prod_{j=1}^{k}dw_{j}+ O(N^{k^2+(2k-2M)n_{1}+2Mn_{2}-1}).
\eeas
Using a similar argument to the analysis of (\ref{516night}), we obtain the claimed result in this theorem.
\end{proof}

To finish, we prove Proposition \ref{intro:prop3}.

\begin{lem}
[Equation (6) of \cite{tauber1963multinomial}]\label{multicombinatorial numbers}
For all $n, m, k_1,\ldots,k_m \in \mathbb{N}$, with
$k_1+\cdots+k_m=n$, $n\geq 1$ and $m\geq 2$, we have
\[
\binom{n}{k_1,\ldots,k_m}
=\sum_{i=1}^m
\binom{n-1}{k_1,\ldots,k_i-1,\ldots,k_m}.
\]
\end{lem}

\begin{proof}[Proof of Proposition \ref{intro:prop3}]
We first show $a_{1,1}(n,0)$. Based on Theorem \ref{intro:thm3} with $k=1$,
it suffices to prove
\beas
(n!)^2\sum_{s=0}^{n}(-1)^s\frac{1}{(n-s)!}\frac{1}{(n+1+s)!}=\frac{1}{2n+1}.
\eeas 
Equivalently,
\beas
\frac{(n!)^2}{(2n+1)!}\sum_{s=0}^{n}(-1)^{n-s}\binom{2n+1}{s}=\frac{1}{2n+1}.
\eeas 
Using Lemma \ref{multicombinatorial numbers} with $m=2$, that is,  
$\binom{2n+1}{s}=\binom{2n}{s}+\binom{2n}{s-1}$, it is easy to see the above equality holds.

In the following, we show $b_{1,1}(n,0)$.
Let $k=M=1$ in Theorem \ref{intro:thm4}, we then have
\beas
&&\int_{\U(N)} |Z_A^{(n)} (1)|^{2}dA_N\nonumber\\
&=& (-1)^{n}
N^{2n+1} (n!)^2 \sum_{i,j=0}^{n}(-1/2)^{2n-i-j}\frac{1}{(n-i)! (n-j)! (i+j+1)!}N^{2n+1}+O(N^{2n})\nonumber \\
&=& (-1)^{n}
N^{2n+1} \frac{(n!)^2}{(2n+1)!} \sum_{i,j=0}^{n}(-1/2)^{i+j}\binom{2n+1}{i,j,2n+1-i-j}N^{2n+1}+O(N^{2n}).
\eeas
Applying Lemma \ref{multicombinatorial numbers} with $m=3$,
\beas
&&\sum_{i,j=0}^{n}(-1/2)^{i+j}\binom{2n+1}{i,j,2n+1-i-j}\\
&=&\sum_{i=1}^{n}\sum_{j=0}^{n}(-1/2)^{i+j}
\binom{2n}{i-1,j,2n+1-i-j} +\sum_{i=0}^{n}\sum_{j=1}^{n}(-1/2)^{i+j}\binom{2n}{i,j-1,2n+1-i-j}\\
&&
+\, \sum_{i=0}^{n}\sum_{j=0}^{n}(-1/2)^{i+j}\binom{2n}{i,j,2n-i-j}\\
&=& 2\sum_{i=0}^{n}\sum_{j=1}^{n}(-1/2)^{i+j}
\binom{2n}{i,j-1,2n+1-i-j} +\sum_{i=0}^{n}\sum_{j=0}^{n-1}(-1/2)^{i+j}\binom{2n}{i,j,2n-i-j} \\
&& + \, \sum_{i=0}^{n}(-1/2)^{i+n}\frac{(2n)!}{n!i!(n-i)!}\\
&=& \sum_{i=0}^{n}(-1/2)^{i+n}\frac{(2n)!}{n!i!(n-i)!}
=\frac{(-1)^{n}(2n)!}{(n!)^22^{n}}\sum_{i=0}^{n}(-1/2)^{i}\binom{n}{i}
=\frac{(-1)^{n}(2n)!}{(n!)^24^{n}}.
\eeas
Putting this all together, we obtain the claimed result in this proposition.
\end{proof}

\section{Numerical data}

The following are $b_{k,k}(2,0)$ for $k=1,\ldots,6$:
\beas
&\displaystyle \frac{1}{2^4\cdot 5}& \\
&\displaystyle \frac{17}{2^{10} \cdot3^{3} \cdot5 \cdot7 \cdot11}& \\
&\displaystyle \frac{11593}{2^{18} \cdot3^{7} \cdot5^{2} \cdot7^{3} \cdot11 \cdot13 \cdot17}& \\
&\displaystyle \frac{103 \cdot413129}{2^{28} \cdot3^{12} \cdot5^{5} \cdot7^{3} \cdot11^{2} \cdot13^{2} \cdot17 \cdot19 \cdot23}& \\
&\displaystyle \frac{2616269 \cdot322433}{2^{40} \cdot3^{17} \cdot5^{8} \cdot7^{5} \cdot11^{4} \cdot13^{3} \cdot17^{2} \cdot19 \cdot23 \cdot29}& \\
&\displaystyle \frac{53 \cdot5830411 \cdot94098709}{2^{54} \cdot3^{24} \cdot5^{13} \cdot7^{8} \cdot11^{4} \cdot13^{4} \cdot17^{3} \cdot19^{2} \cdot23 \cdot29 \cdot31}& \\
\eeas

$a_{2,1}(0,0)=1/12$.

$a_{2,1}(1,n_2)$ for $n_2=0,1$:
\[
{\frac{1}{45}}, {\frac{61}{10080}}.
\]

$a_{2,1}(2,n_2)$ for $n_2=0,1,2$:
\[
{\frac{1}{112}}, {\frac{1133}{453600}}, {\frac{449}{415800}}.
\]

$a_{2,1}(3,n_2)$ for $n_2=0,1,2,3$:
\[
{\frac{1}{225}}, {\frac{529}{415800}}, {\frac{3943}{6879600}}, {\frac{48953}{155232000}}.
\]

$b_{2,1}(0,0)=1/12$.

$b_{2,1}(1,n_2)$ for $n_2=0,1$:
\[
{\frac{1}{720}},{\frac{1}{6720}}.
\]

$b_{2,1}(2,n_2)$ for $n_2=0,1,2$:
\[
{\frac{1}{4032}}, {\frac{19}{3628800}}, {\frac{17}{10644480}}.
\]

$b_{2,1}(3,n_2)$ for $n_2=0,1,2,3$:
\[
{\frac{1}{57600}}, {\frac{19}{10644480}}, {\frac{127}{1761177600}}, {\frac{41}{1419264000}}.
\]

\bibliographystyle{plain}
\bibliography{main}

\end{document}